\documentclass[11pt]{article}
\usepackage{graphicx} %

\usepackage{lineno}

\usepackage{pgfplots}
\usepackage{xcolor}
\usepackage{graphicx}
\pgfplotsset{compat=1.7}
\usetikzlibrary{tikzmark, arrows}
\usepackage{tkz-euclide}
\usetikzlibrary{calc}

\usepackage{natbib}

\usepackage{algorithmic}

\usepackage{enumitem}
\usepackage{float}
\usepackage[figurename=Fig.]{caption}

\usepackage{thmtools} 
\usepackage{thm-restate}

\usepackage[ruled,linesnumbered,vlined,noend]{algorithm2e}
\SetKwComment{Comment}{$\triangleright$\ }{}

\title{Near-optimal quantum query lower bounds on bipartiteness and expansion testing in the bounded-degree graph model}

\date{}

\usepackage[T1]{fontenc}
\usepackage{amsmath,amsthm,amssymb,xcolor}
\usepackage[margin=1in]{geometry}
\usepackage[utf8]{inputenc}
\usepackage{latexsym,bbm,graphicx,float,mathtools}

\usepackage{hyperref}

\usepackage{amsthm}
\usepackage[capitalise]{cleveref}

\usepackage{thm-restate}

\usepackage{apxproof}

\usepackage{authblk}

\usepackage[ruled,linesnumbered,vlined,noend]{algorithm2e}

\usepackage{braket}
\usepackage{stackrel}
\usepackage[all]{xy}

\usepackage{enumitem}
\setlist{itemsep=0em} %

\newcommand{\cA}{\mathcal{A}}

\newcommand{\cP}{\mathcal{P}}

\newcommand{\cO}{\mathcal{O}}

\newcommand{\eps}{\varepsilon}

\usepackage[disable]{todonotes}

\usepackage{todonotes}

\newcounter{mynotes}
\newtheorem{theorem}{Theorem}[section]
\newtheorem{lemma}[theorem]{Lemma}

\newtheorem{proposition}[theorem]{Proposition}
\newtheorem{corollary}[theorem]{Corollary}

\newtheorem{definition}[theorem]{Definition}
\newtheorem{remark}[theorem]{Remark}

\usepackage{abstract}

\usepackage{tikz}
\usetikzlibrary{quantikz2}

\newcommand{\Ftwo}{\mathbb F_2}
\newcommand{\Prb}{\mathbb P}
\newcommand{\YES}{\mathsf{yes}}
\newcommand{\NO}{\mathsf{no}}
\newcommand{\abs}[1]{\left|#1\right|}

\DeclareMathOperator{\fr}{fr}
\DeclareMathOperator{\comp}{comp}
\DeclareMathOperator{\Bin}{Bin}

\definecolor{matchblue}{RGB}{20,112,220}
\definecolor{matchorange}{RGB}{230,105,20}
\definecolor{matchgreen}{RGB}{20,145,58}

\usetikzlibrary{arrows.meta}

\author[1]{Chandrima Kayal\thanks{Email: chandrimakayal2012@gmail.com.}}

\author[2]{Sayantan Sen\thanks{Email: sayantan789@gmail.com.}}

\author[3]{D\'aniel Szab\'o\thanks{Email: szabo.daniel@lmu.de.}}

\affil[1]{Universit\'e Paris Cit\'e, CNRS, IRIF, Paris, France}
\affil[2]{Centre for Quantum Technologies, National University of Singapore, Singapore}
\affil[3]{Ludwig-Maximilians-Universität München \& MCQST, Munich, Germany}

\begin{document}

\maketitle

\begin{abstract}

In this work, we study bipartiteness and expansion testing, two canonical problems in graph property testing in the bounded-degree model through the lens of quantum query complexity. In the classical setting, it is known that $\widetilde{\Theta}(\sqrt{N})$ queries are necessary and sufficient for both these testing problems (Goldreich and Ron, 1999, 2000 \& 2002), where $N$ denotes the number of vertices of the input graph. Due to their significance,   (Ambainis, Childs, and Liu, 2011) initiated the study of these problems in the quantum setting and designed quantum algorithms for bipartiteness and expansion testing that perform $\widetilde{O}(N^{1/3})$ queries, showing a polynomial speedup. They also proved that $\widetilde{\Omega}(N^{1/4})$ queries are necessary for expansion testing, but the possibility of an exponential quantum advantage for bipartiteness testing remained open. Despite significant effort, there has been no improvement in these results in the last decade and a half. 
In this work, we prove essentially tight $\widetilde{\Omega}(N^{1/3})$ quantum query lower bounds for both bipartiteness and expansion testing, thereby completely characterizing the quantum query complexity of these problems up to polylogarithmic factors. 
While our proofs use the polynomial method  similarly to Ambainis, Childs, and Liu, we use intermediate problems that we relate to the main problems via reductions, and perform a more precise analysis of the resulting polynomials, leading to the near-optimal lower bounds.

\end{abstract}

\section{Introduction}
Efficient extraction and inference from large data sets by examining only a part of the input is becoming more and more relevant in the current era of big data. The field of \emph{property testing} designs randomized algorithms that determine whether an object approximately satisfies a given property while accessing only a small part of it.  
When the input object is a graph, it is known as \emph{graph property testing}, where graph properties are closed under graph isomorphism. Formally, an isomorphism-closed collection of graphs $\mathcal{P}$ is called a \emph{graph property}. For a property $\mathcal{P}$, a graph $G$ is said to satisfy $\mathcal{P}$, if $G \in \mathcal{P}$, and is said to be \emph{$\eps$-far} from $\mathcal{P}$, if one needs to modify at least an $\eps$-fraction of the edges in $G$ to obtain a graph in $\mathcal{P}$, for some parameter $\eps \in (0,1)$. 
An \emph{$\eps$-tester} for $\mathcal{P}$ is a randomized algorithm, that accepts graphs $G \in \mathcal{P}$, and rejects graphs that are $\eps$-far from $\mathcal{P}$, with probability at least $2/3$. Over the last two decades, the field of graph property testing has been studied extensively, culminating in a wide array of results and techniques; we refer to the books~\cite{goldreich2017introduction, bhattacharyya2022property} and the surveys~\cite{DBLP:journals/fttcs/Ron09,DBLP:conf/propertytesting/CzumajS10,DBLP:journals/siamdm/RubinfeldS11} for more details.

In graph property testing, there are several access models, depending on the graph representation. In this work, we study the \emph{bounded-degree model}, where the graphs are represented as adjacency lists since the maximum degree of any vertex is at most a constant $d$. The tester can query any neighbor of a vertex of the input graph.

In this work, we focus on \emph{bipartiteness testing} and \emph{expansion testing}, two canonical and well studied problems in this bounded-degree model. In bipartiteness testing, the goal is to distinguish whether the input graph $G$ is bipartite or $\eps$-far from bipartite, by performing as few queries as possible. \cite{goldreich99sublinear,goldreich2002property} proved that \(\widetilde{\Theta}(\sqrt{N})\) queries are necessary and sufficient for bipartiteness testing in this model. Similarly, in expansion testing, the goal is to distinguish whether the input graph $G$ is an $\alpha$-expander, or $\eps$-far from all $\alpha'$-expanders, for some parameters $\alpha, \alpha'$ with $0 < \alpha'  \leq \alpha < 1$. \cite{goldreich2002property, goldreich2011testing} proved that \(\widetilde{\Theta}(\sqrt{N})\) queries are necessary and sufficient for expansion testing as well.~\footnote{More precisely, the best known classical algorithm \cite{Nachmias_shapira_expansion} runs in time $\widetilde O(N^{1/2+\mu})$ and its task is to distinguish $\alpha$-expanders and graphs that are far from $\Omega(\mu\alpha^2)$-expanders; but in the lower bound, $\alpha'$ can be any positive number.}

Quantum algorithms can solve certain problems substantially faster than classical algorithms, as illustrated by Shor’s factoring algorithm~\cite{Shor1994}. This motivated the study of \emph{quantum property testing}~\cite{friedl2003quantum,buhrman2003quantum}. Polynomial quantum speedups are known for several graph-testing problems, including bipartiteness and graph isomorphism testing~\cite{ambainis2011quantum,chakraborty2010new,DBLP:journals/quantum/Apers20,DBLP:conf/tqc/ChenAH24}; see~\cite{montanaro2016survey} for a survey.

A central goal of quantum property testing is either to design quantum testers that outperform all classical testers or to prove lower bounds limiting such improvements. Despite substantial progress, many fundamental questions remain open, and the full extent of the quantum advantage in property testing is still not completely understood.

In the quantum setting of the bounded-degree model, \cite{ambainis2011quantum} studied bipartiteness and expansion testing and designed the first quantum property testers for these problems. In particular, they showed that $\widetilde{O}(N^{1/3})$ quantum queries are sufficient for testing both these properties\footnote{Similarly to the classical case, the expansion testing algorithm distinguishes $\alpha$-expanders and graphs that are far from $\Omega(\mu\alpha^2)$-expanders, and it runs in time $\widetilde O(N^{1/3+3\mu})$.}, thereby demonstrating a polynomial quantum speedup for both these problems. Moreover, they showed that $\widetilde{\Omega}(N^{1/4})$ quantum queries are necessary for expansion testing, thereby showing that the quantum speedup is only polynomial in this case. %
However, they did not prove any lower bound on bipartiteness testing, and the possibility of an exponential quantum speedup for this problem has remained open for over a decade and a half.

Due to its significance, proving a non-trivial quantum query lower bound for this problem was mentioned as an open question in various influential works, such as~\cite{montanaro2016survey, Ben-DavidCGKPW20}. In particular, the authors in \cite{ambainis2011quantum} and \cite{Ben-DavidCGKPW20} explicitly asked the following question:
\begin{center}
    \emph{Is an exponential quantum speed-up possible for bipartiteness testing in the bounded-degree model?}
\end{center}

In other words, the question is whether one can prove any polynomial (in $N$) quantum query lower bound for bipartiteness testing, ruling out the possibility of an exponential quantum speed-up for this problem. 
\color{black}

\subsection{Our results and techniques}
\subsubsection{Main results}

In this work, we settle (up to polylogarithmic factors) the quantum query complexities of both bipartiteness testing and expansion testing, proving that $\widetilde{\Omega}(N^{1/3})$ quantum queries are indeed necessary. This demonstrates that the quantum property testing algorithms for bipartiteness and expansion testing from \cite{ambainis2011quantum} are essentially optimal.

\begin{theorem}[Informal statement of \Cref{theo:biplb}]\label{thm:informal}
Any two-sided error quantum $\eps$-tester for bipartiteness testing on $N$ vertex bounded-degree graphs is required to perform $\widetilde{\Omega}(N^{1/3})$ queries.
\end{theorem}

\begin{theorem}[Informal statement of \cref{theo:expansionlb}]\label{thm:informal-expansion}
Any two-sided error quantum $\eps$-tester for testing expansion of $N$-vertex bounded-degree graphs must make $\widetilde{\Omega}(N^{1/3})$ queries.
\end{theorem}

\subsubsection{Technical overview}
Here we present the main ideas in the proofs of our results. Since we rely on some parts of \cite{ambainis2011quantum}, we first recap their proof idea.

\paragraph{The lower bound of \cite{ambainis2011quantum}:}
We give a short overview of the lower bound on expansion testing from \cite{ambainis2011quantum} to see which ideas we reuse and where we diverge from their proof.
They consider the following hard distribution $P_{M,\ell}$ of input graphs: take $M$ vertices and consider a partition of the vertex set into $\ell$ blocks of equal size $V_1,\dots,V_\ell$ (we take the parameters such that $2\ell$ divides $M$). In each $V_i$ for $i\in[\ell]$, take $c$ random perfect matchings, and assign a color to each matching. This way, every vertex has exactly one incident edge of each color. The graph constructed this way has maximum degree $c$, and we will call it the \emph{ambient} graph $A$.

The actual input graph $G$ on $N$ vertices ($N\le M$) is constructed by sampling vertices randomly from $A$. For each vertex of $G$, first sample a uniformly random block $V_i$ and then sample a uniformly random vertex inside $V_i$ among those that have not been sampled yet. $G$ is the subgraph of $A$ induced by the sampled vertices, thus it has maximum degree at most $c$. One can prove that the ``overflow'' event, i.e. when more than $M/\ell$ vertices are sampled from the same block, is very unlikely.

In their proof they take $M=(1+\Theta(N^{-0.1}))N$. They show that when the parameter $\ell=1$, then graph $G$ is an expander with high probability, and it is clear that $\ell=2$ always results in a far from expander graph, since it is not connected.

Then the main idea of the proof is the following. Graphs sampled from $P_{M,1}$ should be accepted with high probability, and those sampled from $P_{M,2}$ should be accepted with small probability; moreover, the acceptance probability is trivially bounded by 1 for all valid parameter pairs $(M,\ell)$. This suggests to use the \emph{polynomial method}: proving a lower bound on the degree of the acceptance probability (as a polynomial) yields a lower bound on the quantum query complexity. 

In \cite[Lemma 3]{ambainis2011quantum} it is proved that the expected acceptance probability of a $q$-query quantum algorithm over distribution $P_{M,\ell}$ is a fraction of two bivariate polynomials $\frac{f(M,\ell)}{g(M,\ell)}$, where both $f$ and $g$ are of degree $O(q\log q)$ and the denominator divides $D_q=\prod_{k=1}^{2q}(M-(2k-1)\ell)^{\lfloor 2q/k\rfloor}$.

In the proof they consider the acceptance probability polynomial and examine each monomial separately, in expectation. Each monomial is a product of indicator variables $x_{u,v,j}$ of the event ``vertices $u$ and $v$ are connected by an edge of color $j$''. This way, each monomial defines a graph pattern, and the expected value of the monomial is the probability that the pattern appears in the graph sampled from $P_{M,\ell}$. They analyze these expected values as rational functions, including the number of $\ell$-factors in the numerator.

To finish the lower bound, they bound the denominator of the expected acceptance probability by a constant, to obtain a simple polynomial rather than a fraction. Then they prove a lemma using a similar argument to the collision lower bound of \cite{aaronson2004quantum}, to lower bound the degree of the polynomial. For this, they consider parameter grid $|M-M_0|\le R^{3/2}$ and $\ell\le R$ with $M_0=(1+\Theta(N^{-0.1}))N$ and $R\in\Theta(N^{1/2})$, such that $2\ell$ divides $M$. Their lemma yields a degree lower bound of $\Omega(\sqrt{R})=\Omega(N^{1/4})$, and since the degree is $O(q\log q)$, the query lower bound is $\widetilde\Omega(N^{1/4})$.

\paragraph{Bipartiteness testing lower bound}
First, we look at a property on \emph{signed graphs} that we call \emph{balancedness} and prove that property testing this problem has the desired lower bound. Then we reduce it to bipartiteness testing and thus the quantum query lower bound extends to bipartiteness testing as well. 
A signed graph is a graph with bit labels (also called signs) on its edges, and it is balanced if there is an assignment of bits to the vertices such that every edge $uv$ has label that is the XOR of the bits of $u$ and $v$. 

For the lower bound, we begin with the hard instance construction of \cite{ambainis2011quantum} described in the previous section. We first consider a graph defined in terms of $c=3$ matchings over blocks of vertices. Any graph in the hard distribution has at least $1.3N$ edges with high probability. Once we have the base graph similar to \cite{ambainis2011quantum}, we define a signed variant of it in the following way. 
In the $\YES$ distribution, first a vertex signing is taken and the edge signs are defined accordingly: these are clearly balanced. In the $\NO$ distribution, the edge signs are taken uniformly at random, which results in far from balanced graphs with high probability.

Using \cite[Lemma 3]{ambainis2011quantum}, we know that the average acceptance probability of a $q$-query quantum query algorithm with inputs sampled from the hard distribution $P_{M,\ell}$, is a fraction of two polynomials each of degree $O(q \log(q))$, and the denominator divides $D_q=\prod_{k=1}^{2q}(M-(2k-1)\ell)^{\lfloor 2q/k\rfloor}$. We denote this acceptance probability as $p_{\YES}$ and $p_{\NO}$ when taking the edge signs according to the $\YES$ and the $\NO$ distribution, respectively.\footnote{
There is a subtle difference between the actual distribution (that we will sometimes call formal) and the one conditioned on no overflow event happening (that we call feasible). However, we will show that the overflow probability is small, and thus we disregard it in this section.}

Like \cite{ambainis2011quantum}, we also look at each monomial, which all correspond to a signed graph pattern. If this pattern is a forest, then the probabilities of this pattern appearing in a graph from the $\YES$ and the $\NO$ distribution are the same. This way, when considering the acceptance gap $\Delta(M,\ell)=p_{\YES}(M,\ell)-p_{\NO}(M,\ell)$ over the common denominator $D_q$, the monomials of forests cancel, and only ones corresponding to graphs with cycles remain.

By carefully analyzing the proof of \cite[Proposition 6]{ambainis2011quantum}, one can see that a better bound number of $\ell$-factors in the numerator of the monomials is $\ell^{e-v+c}$ where $e,v,c$ denote the number of edges, vertices and components of the corresponding graph pattern, respectively; and $e-v+c$ is the cycle rank of the graph pattern. Since the forests are already canceled in the acceptance gap, each remaining monomial brings at least one $\ell$-factor, which allows us to write $\Delta(M,\ell)=\frac{\ell H(M,\ell)}{D_q(M,\ell)}$.

The next step is to approximate the reciprocal of the denominator by a low degree polynomial. This is done by writing $D_q$ in a special product form, rewriting its reciprocal as an infinite sum and truncating it. The resulting polynomial $A_q(M,\ell)$ satisfies $|D_qA_q-1|\le0.01$. 

Now we define our final polynomial as $P(M,\ell)=H(M,\ell)A_q(M,\ell)/2$ which still satisfies $\deg(P)\in O(q\log q)$. 
We consider the hard distribution over the enlarged grid $|M-M_0|\le R^{3/2}$ and $\ell\le R$ with $M_0=4\lceil \frac{11N}{40} \rceil$ and $R\in\Theta(N^{2/3})$, such that $2\ell$ divides $M$.
Moreover, we know that $P(M_0,1)$ and $P(M_0,2)$ are bounded away by a constant and $P$ is bounded by 1 over the enlarged grid. Thus, using \cite[Lemma 4]{ambainis2011quantum}, we can prove the degree lower bound of $\Omega(\sqrt{R})=\Omega(N^{1/3})$ which results in a query lower bound of $\widetilde\Omega(N^{1/3})$.

To obtain our first main result, we prove a reduction from the balancedness property testing problem to bipartiteness testing. The reduction replaces each 0-label edge with a length-2 path and each 1-label edge with a length-3 path. The reduction works because of the fact that the distance of a signed graph from balancedness is the minimum number of edges whose deletion makes the signed graph balanced.

\begin{figure}
\centering
\begin{tikzpicture}[
  font=\small,
  vertex/.style={circle,draw=black,fill=white,minimum size=7mm,
    inner sep=0pt,line width=.7pt},
  edge/.style={line width=1.25pt}
]

\begin{scope}[xshift=0cm]
  \node[vertex] (yu) at (0,1.65) {$u$};
  \node[vertex] (yv) at (3.00,1.65) {$v$};
  \node[vertex] (yw) at (0,0) {$w$};
  \node[vertex] (yz) at (3.00,0) {$z$};

  \draw[edge,matchblue]   (yu)--(yv);
  \draw[edge,matchorange] (yu)--(yw);
  \draw[edge,matchorange] (yw)--(yz);
  \draw[edge,matchgreen]  (yv)--(yz);

  \node[anchor=south east,font=\scriptsize] at (yu.north west) {$\xi_u$};
  \node[anchor=south west,font=\scriptsize] at (yv.north east) {$\xi_v$};
  \node[anchor=north east,font=\scriptsize] at (yw.south west) {$\xi_w$};
  \node[anchor=north west,font=\scriptsize] at (yz.south east) {$\xi_z$};

  \node[anchor=north,font=\scriptsize] at (1.50,-.72)
    {$\xi_x\stackrel{\mathrm{i.i.d.}}{\sim}
      \operatorname{Bernoulli}(1/2),\qquad
      b_{xy}=\xi_x\oplus\xi_y$};
\end{scope}

\begin{scope}[xshift=7.25cm]
  \node[vertex] (nu) at (0,1.65) {$u$};
  \node[vertex] (nv) at (3.00,1.65) {$v$};
  \node[vertex] (nw) at (0,0) {$w$};
  \node[vertex] (nz) at (3.00,0) {$z$};

  \draw[edge,matchblue]   (nu)--(nv);
  \draw[edge,matchorange] (nu)--(nw);
  \draw[edge,matchorange] (nw)--(nz);
  \draw[edge,matchgreen]  (nv)--(nz);

  \node[anchor=south,font=\scriptsize] at (1.50,1.72) {$b_{uv}$};
  \node[anchor=north,font=\scriptsize] at (1.50,-.07) {$b_{wz}$};
  \node[rotate=90,font=\scriptsize] at (-.35,.825) {$b_{uw}$};
  \node[rotate=90,font=\scriptsize] at (3.35,.825) {$b_{vz}$};
  \node[anchor=north,font=\scriptsize] at (1.50,-.72)
    {$b_e\stackrel{\mathrm{i.i.d.}}{\sim}
      \operatorname{Bernoulli}(1/2)\quad(e\in E(H))$};
\end{scope}
\end{tikzpicture}
\caption{On the left, we have a $\YES$ instance balanced graph, and on the right, the $\NO$ instance graph illustrates independent random edge signs. \Cref{lem:random-sign-farness} shows that the $\NO$ instance is far from balanced with high probability.}
\end{figure}
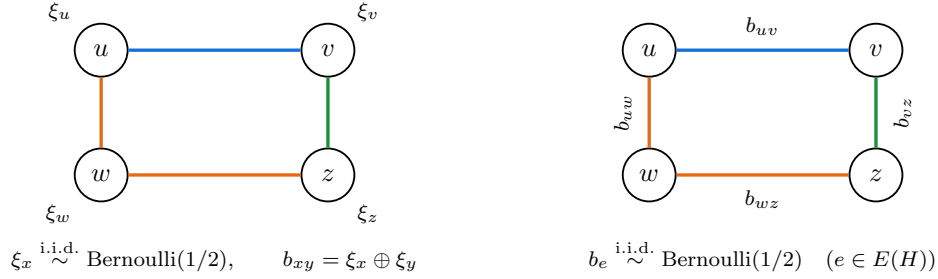

\begin{figure}
\centering
\begin{tikzpicture}[
  font=\small,
  original/.style={circle,draw=black,fill=white,minimum size=7mm,
    inner sep=0pt,line width=.7pt},
  auxiliary/.style={circle,draw=black,fill=white,minimum size=4.7mm,
    inner sep=0pt,line width=.7pt},
  path/.style={draw=black,line width=.9pt},
  mapsto/.style={-{Latex[length=2.6mm,width=1.8mm]},draw=black,
    line width=.8pt}
]

\begin{scope}[xshift=0cm]
  \node[original] (ou0) at (0,0) {$u$};
  \node[original] (ov0) at (1.80,0) {$v$};
  \draw[path] (ou0)--node[above,font=\scriptsize] {$b_{uv}=1$} (ov0);
  \draw[mapsto] (2.30,0)--(3.10,0);

  \node[original]  (ou)  at (3.50,0) {$u$};
  \node[auxiliary] (opu) at (4.50,0) {};
  \node[auxiliary] (opv) at (5.50,0) {};
  \node[original]  (ov)  at (6.50,0) {$v$};
  \draw[path] (ou)--(opu)--(opv)--(ov);
  \node[above=2pt,font=\scriptsize] at (opu.north) {$a_{u,j}$};
  \node[above=2pt,font=\scriptsize] at (opv.north) {$p_{u,j}$};
  \node[anchor=north,font=\scriptsize] at (5.00,-.45)
    {$\ell(P_{uv})=3$};
\end{scope}

\begin{scope}[xshift=8.50cm]
  \node[original] (eu0) at (0,0) {$u$};
  \node[original] (ev0) at (1.80,0) {$v$};
  \draw[path] (eu0)--node[above,font=\scriptsize] {$b_{uv}=0$} (ev0);
  \draw[mapsto] (2.30,0)--(3.10,0);

  \node[original]  (eu)  at (3.50,0) {$u$};
  \node[auxiliary] (epu) at (5.20,0) {};
  \node[original]  (ev)  at (6.90,0) {$v$};
  \draw[path] (eu)--(epu)--(ev);
  \node[above=2pt,font=\scriptsize] at (epu.north) {$a_{u,j}$};
  \node[anchor=north,font=\scriptsize] at (5.20,-.45)
    {$\ell(P_{uv})=2$};
\end{scope}
\end{tikzpicture}
\caption{We describe the reduction from balancedness testing to bipartiteness testing. On the left, the construction describes what happens when $b_{uv}=1$ and on the right, it describes the reduction when $b_{uv}=0$.}
\end{figure}
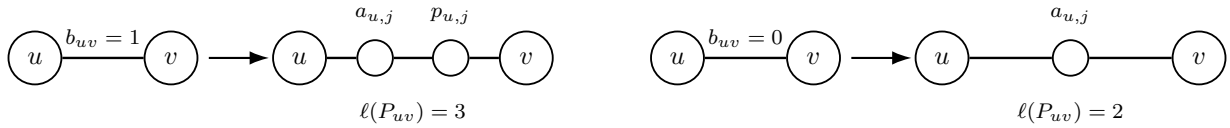

\paragraph{Expansion testing lower bound}

Our lower bound proof for expansion testing follows a similar proof recipe as bipartiteness testing. But there are some important differences, which we discuss below.

For expansion testing, we consider similar hard instances as bipartiteness testing. However, using the same enlarged grid and $c=3$ number of matchings will no longer work since any vertex in the final graph has a constant probability of losing all of its matching neighbors. This would make the graphs in the $\YES$ distribution no longer expanders. To bypass this bottleneck, we consider $c=\Theta(\log N')$ perfect matchings in the input graph instead of three perfect matchings. This ensures that with high probability, in a $G$ of size $N'$ sampled from the $\YES$ distribution, every cut will have sufficient expansion, even when a constant fraction of the ambient vertices is unused. This ensures that $G$ is an expander.

However, adding $\Theta(\log N')$ perfect matchings will no longer keep the graph with constant maximum degree, which is contrary to the definition of the bounded degree model. This can be fixed by a reduction: in $G$, that has $N'$ vertices, we introduce a replacement graph gadget of $\Theta(\log N')$ vertices for each vertex in the intermediate graph. Combining the graph gadget with the intermediate graph, we obtain a graph $H$ on $N= \Theta(N' \log N')$ vertices and with maximum degree $4$, which satisfies our assumption of constant maximum degree.

We show that if $G$ is an expander then $H$ is as well, and similarly for being far from expanders. Thus, using a simulation argument, if expansion of $H$ is testable in $q$ queries then that of $G$ too, so it suffices to show a lower bound for $G$. This is done by the same argument as for balancedness testing. Because of the different sizes of $G$ and $H$, the polylogarithmic factor in the final bound is different than for bipartiteness, nonetheless we still obtain an $\widetilde\Omega(N^{1/3})$ quantum query lower bound.

\paragraph{Conclusion}
In this paper we prove  near-optimal quantum query complexity lower bounds for bipartiteness and expansion property testing. Our main tool is the framework of \cite{ambainis2011quantum} that uses the polynomial method.
The main new ideas for proving the near-optimal lower bounds are characterizing the exact number of $\ell$-factors in the numerator which lets us divide by the block parameter; and the use of a reciprocal polynomial instead of replacing the denominator by a constant. We believe that these ideas are potentially useful beyond the applications in this work.

\subsection{Open problems}
One problem left open by this work is determining the exact power of $\log(N)$ in the complexities of both bipartiteness and expansion testing.

Another interesting question is how the complexity of these problems depends on the property testing distance parameter $\eps$, i.e., what if we let $\eps$ depend on $N$ rather than considering it a constant.

Similarly, the dependence of the query complexity of expansion testing on the expansion parameter could be further examined. Note that both the classical and quantum algorithms distinguish $\alpha$-expanders and far from $\alpha'$-expanders, and they achieve the near-optimal complexity only if $\alpha'$ is significantly smaller than $\alpha$. It would be interesting to see how the complexity changes for different $\alpha$ versus $\alpha'$ gaps.

\paragraph{Organization}
In \Cref{sec:prelim}, we introduce the necessary notations, definitions and previous results.
In \Cref{sec:bip_lb}, we prove a lower bound on the balanced signed graph testing problem and reduce it to bipartiteness testing, thus yielding our first main result. 
In \Cref{sec:expansiontesting}, our second main result is proved: the lower bound on expansion testing. The proofs of some of the lemmas are presented in the appendix.

\section{Preliminaries}\label{sec:prelim}

\paragraph{Notation:} $[n]$ denotes the set $\{1, \ldots, n\}$, and $\oplus$ denotes the Boolean XOR operation. $h_2(\cdot)$ denotes the binary entropy function. 
For simplicity, for an edge $e=\{u,v\}$ of a graph we will sometimes use the shorthand $e=uv$.
Throughout this work, we will use the notations $O(\cdot)$ and $\Omega(\cdot)$ to hide the dependencies on parameters $\eps$ and $d$ that we consider to be constants. We will also use $\widetilde{O}(\cdot)$ and $\widetilde{\Omega}(\cdot)$, where we hide poly-logarithmic dependencies on the parameters. 
We assume that the reader has some familiarity with quantum computing and refer interested readers to the textbook \cite{nielsen_chuang} for further background.

Let us start by describing our query models.

\subsection{Query models}
\subsubsection{Classical Query Model}
In classical query complexity, we consider inputs $x\in \Sigma^I$ over a finite alphabet $\Sigma$ and indexed by a set $I$.
However, the input is not given explicitly to the algorithm. Instead, the algorithm has query access to an input oracle $\cO_x:I\to\Sigma$ encoding $x$ by $\cO_x(i)=x_i$, where the $i$-th entry of $x$ is accessed. 

\emph{Query complexity} measures the minimum number of queries that an algorithm has to make in order to decide whether an unknown input satisfies a property $\cP$ of interest. 

In this work, we will focus on the undirected bounded-degree graph model. In this model, the algorithm has query access to the adjacency list of an undirected bounded-degree graph $G$ with $n$ vertices and maximum degree $d$ (a constant). The query access to $G$ is represented as an oracle $\cO_G: V \times [d] \rightarrow V \cup \{\bot\}$. Then we can set $I=V\times [d]$ and $\Sigma=V\cup\{\bot\}$,  such that for any $v \in V$ and $i \in [d]$, we have the following:
    $$\cO_G(v,i)=
\begin{cases}
    w , & \text{if $w\in V$ is the $i$-th neighbor of $v$};\\
    \bot, & \text{if deg}(v) <i.
\end{cases}
$$    

\subsubsection{Quantum Query Model} \label{sec:quantum_queries}
In quantum query complexity, the input oracle is a unitary $O_x$ over Hilbert space $\mathcal{H}_{I}\otimes \mathcal{H}_{\Sigma}$ defined as $O_x\ket{i,y}=\ket{i,y \oplus x_i}$, for $y\in \Sigma$ and $i\in I$, where $\oplus$ denotes the bit-wise XOR operation up to some binary encoding of the elements of $\Sigma$.

In the article, we are going to consider the setting where we have adjacency list query access to a bounded degree graph $G$. Using the above formalism, the oracle unitary can be written as $O_G\ket{v,i,y}=\ket{v,i,y \oplus \cO_G(v,i)}$. Here $v, y\in V(G)\cup\{\bot\}$ and $i\in[d]$.

In quantum query complexity, we are only interested in the number of times an algorithm uses the unitary oracle, thus all the input-independent unitaries it performs between any two oracle calls can be written as a single unitary. This way, the state of a $T$-query quantum algorithm after an arbitrary $t\le T$ number of oracle calls can be written as $\ket{\psi_t^x}=U_tO_xU_{t-1}\dots U_1O_xU_0\ket{0}$. 
Here, we assume that the algorithm's state consists of three registers: the query register $Q$ holds the queried index, the answer register $Y$ holds the query outcome, and there is some additional workspace register $W$ the algorithm can use. This way, $\ket{\psi_t^x}$ is a superposition of basis states $\ket{i,y,w}_{QYW}$.

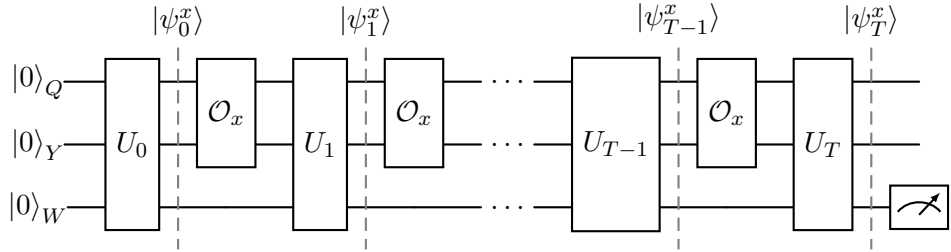
\begin{figure}[h]
	\centering
	\begin{quantikz}
		\ket{0}_Q &\gate[3]{U_0} \slice[style=gray]{$\ket{\psi_0^x}$} &\gate[2]{\cO_x}& \gate[3]{U_1} \slice[style=gray]{$\ket{\psi_1^x}$} &\gate[2]{\cO_x}& \ \ldots\ & \gate[3]{U_{T-1}} \slice[style=gray]{$\ket{\psi_{T-1}^x}$} & \gate[2]{\cO_x}& \gate[3]{U_T} \slice[style=gray]{$\ket{\psi_T^x}$}& \\[-0.75em] 
		\ket{0}_Y&&&&& \ \ldots\  &&&& \\[-0.75em] 
		\ket{0}_W&&&&& \ \ldots\  &&&& \meter{}
	\end{quantikz}
	\caption{The circuit of a quantum query algorithm.}
	\label{fig:circuit_query}
\end{figure}

At the end of the computation, we assume that the algorithm's output is written on a dedicated part $W_{out}$ of register $W$. In the final state $\ket{\psi_T^x}$, register $W$ is measured and the algorithm succeeds if the value $w_{out}$ seen at $W_{out}$ is a correct output for input $x$.

The above representation can be used to prove the polynomial method (see \Cref{lem:beals}) for proving quantum query complexity lower bounds.

\subsection{Property testing}

We now formally give the definitions related to property testing. Let us start with the notion of a graph property and what we mean by being far from a graph property.

\begin{definition}[Graph property]
A \emph{graph property} is a collection of graphs that are closed under permutation of the vertices. More formally, $\cP$ is a graph property if and only if for every graph $G=([n], E)$ and every permutation $\sigma$ of $[n]$, $G \in \cP$ if and only if $\sigma(G) \in \cP$, where $\sigma(G)=([n], \{\{\sigma(u), \sigma(v)\}: \{u,v\} \in E\})$. 
\end{definition}

\begin{definition}[$\eps$-far]
    An $N$-vertex graph $G$ with bounded degree $d$ is said to be \emph{$\eps$-far} from $\cP$ (for some $\eps\in(0,1)$) if one needs to modify (either add or remove) at least $\eps N d$ edges from $G$ such that the modified graph will have the property $\cP$.
\end{definition}

Now we formally define an $\eps$-tester.

\begin{definition}[$\eps$-tester]
Let $G=(V,E)$ be an unknown bounded degree graph with $n$ vertices with maximum degree $d$ and $\eps \in (0,1)$ be a parameter. Moreover, let $\cP$ be a graph property. Given query access to the adjacency list of $G$, an $\eps$-tester for $\cP$ is an algorithm $\cA$ that can distinguish between the two following cases with probability at least $2/3$:
\begin{itemize}
    \item[(i)] If $G \in \cP$, then $\cA$ outputs accept.

    \item[(ii)] If $G$ is $\eps$-far from $\cP$, then $\cA$ outputs reject.
\end{itemize}
The tester $\cA$ can either output accept or reject when the above two conditions do not hold. %
\end{definition}

\begin{remark}
Any $\eps'$-tester by definition is an $\eps$-tester for any $0< \eps' \leq \eps <1$. Similarly, whenever in this work we say any $\eps$-tester for some property $\mathcal{P}$ requires $q$ queries, we mean that there exists some constant parameter $0<\eps' <1$ such that any $\eps$-tester with $0< \eps \leq \eps'$ requires $q$ queries.    
\end{remark}

For expansion testing, throughout this work, we will consider the notion of vertex expansion.

\begin{definition}[Vertex expansion]\label{def:vertex-expansion}
Let $G=(V,E)$ be an $N$-vertex graph. For $U\subseteq V$, let
\[
 \partial_G(U)=\set{v\in V\setminus U:\text{$v$ has a neighbor in $U$}}.
\]
For $\alpha>0$, $G$ is said to be an $\alpha$-expander if
\[
 \abs{\partial_G(U)}\geq\alpha\abs U
\]
for every $U\subseteq V$ with $\abs U\leq N/2$.
\end{definition}

Throughout this work, unless mentioned explicitly, whenever we mention ``expansion'', it corresponds to vertex expansion. Some arguments below first
count edges crossing between a set $S$ and its complement $V\setminus S$. The set of edges like that are denoted as $e(S,V\setminus S)$. In a graph of maximum degree
$d$, such a bound on the number of edges implies the corresponding vertex-boundary
bound after a loss of at most a factor $d$.

\begin{remark}
Any $\alpha'$-expander by definition is an $\alpha$-expander for any $0< \alpha \leq \alpha' <1$. Similarly, whenever in this work we say that an algorithm that distinguishes $\alpha$-expanders from graphs that are $\eps$-far from being $\alpha'$-expanders needs to make $q$ queries, we mean that there exists some constant parameter $0<\alpha_0 <1$ such that for any $0< \alpha\leq \alpha' \leq \alpha_0$ and $\eps$ at most some function of $\alpha$, the testing problem requires $q$ queries.   
\end{remark}

\subsection{Signed graphs, balance, and frustration} \label{sec:signed_graphs}

A signed graph is a graph $H=(V,E)$ together with a bit $b_e\in\Ftwo$ on every edge. We will call $b_e$ the label or the sign of edge $e$. In this work, parallel edges are allowed 
and are treated as distinct constraints.

\begin{definition}[Balance and frustration]
A signed graph $(H,b)$ is balanced if there is a vertex labeling $x\in\Ftwo^V$ such that
\begin{equation}\label{eq:balance}
 x_u\oplus x_v=b_{uv}
\end{equation}
for every edge $uv\in E$. The frustration index of $H$ with edge sign vector $b$ is
\begin{equation}\label{eq:frustration}
 \fr(H,b)=\min_{x\in\Ftwo^V}
 \abs{\set{uv\in E:x_u\oplus x_v\neq b_{uv}}}.
\end{equation}
\end{definition}

\begin{remark}
    Note that in the literature it is also common to consider edge and vertex labels from $\{+1,-1\}$, and in this case balancedness requires $x_ux_v=b_{uv}$ for every edge $uv\in E$.
\end{remark}

A signed graph is balanced if and only if every cycle has an even number of edges with label $1$~\cite{Harary53}. The frustration index is the 
minimum number of edges whose deletion makes the signed graph balanced (see e.g., \cite{Zaslavsky18}).

For signed graphs, we will use an oracle that, on querying $(v,i)$, returns either $\bot$ (if $v$ has no $i$-th neighbor) or $(w,b)$ where $w$ is the $i$-th neighbor of $v$ and the label of edge $vw$ is $b$. In particular, in our hard instance $i\in\{1,2,3\}$ will be the id of the matching the edge belongs to, that we also call color. In our hard instances, we will consider graphs that are union of matching edges, each with a different color. In such a graph, parallel edges are possible, but only of different colors. We call such graphs \emph{colored multigraphs}.

\subsection{The quantum acceptance polynomial}
\label{sec:acc_poly}
For every possible non-null outcome $z=(w,b)$ to a query $a=(v,i)$, let us introduce the indicator $X_{a,z}=\mathbbm{1}_{O(a)=z}$.
The null-outcome indicator can be replaced by $X_{a,\perp}=1-\sum_{z\neq\perp}X_{a,z}$.
We reduce the resulting polynomial using $X_{a,z}^2=X_{a,z}$ and $X_{a,z}X_{a,z'}=0$ for any $z\neq z'$.
Thus a nonzero monomial contains at most one response indicator for each possible query.

\begin{lemma}[\cite{beals2001quantum}]\label{lem:beals}
The acceptance probability of a quantum algorithm making $q$ oracle queries is a real multilinear polynomial of degree at most $2q$ in the response indicators.
\end{lemma}

Every surviving monomial specifies a partial response table. In our setting of the bounded-degree graph model, each consistent table determines a colored multigraph pattern $F$ (in the case of signed graphs, together with prescribed signs on its distinct edges). Inconsistent monomials have expectation zero and may be discarded. For the number of edges, vertices and components of $F$, respectively, we write
\[
 e(F)=\abs{E(F)},\qquad
 v(F)=\abs{V(F)},\qquad
 c(F)=\comp(F),
\]
where only vertices incident to a specified edge are counted. By the well-known Euler formula, the cycle rank (i.e., the minimum number of edges to be deleted for removing all cycles) is
\begin{equation}\label{eq:cycle-rank}
 \beta(F)=e(F)-v(F)+c(F).
\end{equation}
It is zero if and only if $F$ is a forest.

We will also use a standard coefficient estimate. A proof is included in \Cref{app:proof_2.8} for completeness.

\begin{restatable}%
{lemma}{lemcoefficientbound}\label{lem:coefficient-bound}
Consider the acceptance polynomial in the canonical form described above. A coefficient of a degree-$k$ monomial has absolute value at most $2^k$. Consequently, if the number of possible query outcomes is $O(N)$ and the algorithm makes $q$ queries, then the sum of the absolute values of all coefficients is at most $\exp(O(q\log N))$.
\end{restatable}

\subsection{The Ambainis--Childs--Liu distribution and lemmas}

\color{red}

\color{black}
As we will use the hard distribution and some results from \cite{ambainis2011quantum}, we describe these here.
\subsubsection{The hard distribution from \cite{ambainis2011quantum}} \label{sec:hard_instance_acl}
We first describe the hard distribution $P_{M,\ell}^{(c)}$ of graphs $G$ on $N$ vertices  from \cite{ambainis2011quantum}, that they used to prove lower bound for expansion testing. Let $M\geq N$ and suppose that $2\ell$ divides $M$. We first describe the \emph{ambient} graph $A$ on $M$ vertices, and our graph $G$ will be an induced subgraph of $A$. We will call the vertices of $G$ \emph{visible}.
\begin{enumerate}
\item Partition a set of $M$ vertices into $\ell$ blocks, each containing $M/\ell$ vertices.

\item Sample $c$ uniformly random perfect matchings independently in every block, one for each color. This defines the ambient graph $A$.

\item Each one of the $N$ vertices of $G$ is sampled in the following way. 
Choose a block among the $\ell$ blocks uniformly and independently at random. Within the chosen block, sample the visible vertex uniformly at random from those ambient vertices that have not been sampled yet. After all the $N$ vertices are sampled, $G$ is defined as the subgraph of $A$ induced by the visible vertices.%
\end{enumerate}

If the number of colors $c$ is clear from the context, we will just write $P_{M,\ell}$ rather than $P_{M,\ell}^{(c)}$.
If some block is assigned more than $M/\ell$ visible vertices, the process fails, we call this an \emph{overflow} event. However, this event is unlikely, we prove this in \cref{lem:overflow}. We call distributions conditioned on the overflow event not happening as \emph{feasible distributions}.
The resulting visible graph has degree at most $c$, with at most one incident edge of one color. %

\subsubsection{Two important lemmas}
Writing the quantum acceptance polynomial from \Cref{sec:acc_poly}, one can calculate the probability of each monomial (graph pattern) over the hard distribution $P_{M,\ell}^{(c)}$. This way, \cite{ambainis2011quantum} proved the following lemma about the expected acceptance probability of an algorithm.

Before stating the lemma, let us define the common denominator
\begin{equation}\label{eq:denominator}
 D_q(M,\ell)
 =\prod_{a=1}^{2q}
 \bigl(M-(2a-1)\ell\bigr)^{\lfloor 2q/a\rfloor}.
\end{equation}
Its degree is
\begin{equation}\label{eq:den-degree}
 d_q:=\deg D_q
 =\sum_{a=1}^{2q}\left\lfloor\frac {2q}{a}\right\rfloor
 =\Theta(q\log q).
\end{equation}

We will consider the expected acceptance probability with denominator $D_q$. This does not change the degree of the denominator or the numerator much, they both stay $O(q\log q)$.

\begin{lemma}[{\cite[Lemma 3]{ambainis2011quantum}}]\label{lem:acl-rational}
The expected acceptance probability of a $q$-query quantum algorithm under $P_{M,\ell}^{(c)}$ is a rational function whose denominator divides $D_q(M,\ell)$ and whose numerator has degree $O(q\log q)$. The same statement holds for any linear combination of expected acceptance probabilities.
\end{lemma}

The proof in~\cite{ambainis2011quantum} expands the degree-$2q$ acceptance polynomial into monomials, evaluates each graph pattern probability, and puts the probabilities over a common denominator. In the case of balancedness testing in \Cref{sec:bip_lb}, the extra sign bit in the query outcome changes neither the graph pattern probability nor their degree.

We also use the bivariate polynomial lemma from \cite{ambainis2011quantum}. We state a translated and rescaled version; it follows from the same proof, which uses Paturi's theorem~\cite{paturi1992degree}.

\begin{lemma}[{\cite[Lemma 4]{ambainis2011quantum}}]\label{lem:acl-degree}
Let $C$ and $R$ be positive integers with $R\geq2$. Suppose that $Q(X,\ell)$ is a real polynomial such that
$\abs{Q(X,\ell)} \leq 1$
whenever $1\leq\ell\leq R, \abs{X-C}\leq\frac{R^{3/2}}2, \ell\mid X$.
If $\abs{Q(C,1)-Q(C,2)}\geq\gamma
$ for a fixed constant $\gamma>0$, then
$\deg Q=\Omega(\sqrt R)$
where the implicit constant depends only on $\gamma$.
\end{lemma}

\begin{remark}
The statement in~\cite{ambainis2011quantum} sets $C=aN$. Its proof uses only the arithmetic grid around the center and therefore gives the translated statement above. The proof considers the polynomial first along the line $X=C$ and, if its values become large there, along an arithmetic progression in the $X$ direction. Paturi's theorem gives degree $\Omega(\sqrt R)$ in either case.
\end{remark}

\section{Quantum query lower bound for bipartiteness testing} \label{sec:bip_lb}

\begin{restatable}[Formal version of \cref{thm:informal}]{theorem}{theobiplb}\label{theo:biplb}
Let $G$ be a bounded-degree graph on $N$ vertices and $\eps \in (0,1)$ be a constant. In order to distinguish with probability at least $2/3$ if $G$ is bipartite, or $\eps$-far from being bipartite, $\Omega(N^{1/3}/\log N)$ quantum queries are necessary. 
\end{restatable}

In order to prove the above theorem, we first consider balancedness testing in signed bounded degree graphs, and prove that balancedness testing is hard in quantum setting. Then we will reduce it to bipartiteness testing. Our result for balancedness testing is as follows.

\begin{restatable}%
{theorem}{balancedtest}\label{theo:balanced}
Let $G$ be a bounded-degree signed graph on $N$ vertices and $\eps \in (0,1)$ be a constant. In order to distinguish with probability at least $2/3$ if $G$ is balanced or $\eps$-far from being balanced, $\Omega(N^{1/3}/\log N)$ quantum queries are necessary.    
\end{restatable}

We will define the notion of farness from being balanced via the frustration index (see \cref{eq:frustration}).

\subsection{Description of our hard instances}
\label{sec:newconstruction}

For the hard instance of our lower bound for balancedness testing, we will use the construction of \cite{ambainis2011quantum} that we described in \Cref{sec:hard_instance_acl} setting the number of perfect matchings (colors) to $c=3$.
Let us call the graph obtained by the random process $P_{M,\ell}$ a \emph{base} graph. Once we obtain a base graph, we define an associated signed graph (see \Cref{sec:signed_graphs}) by assigning a bit $b_{uv}\in \Ftwo$ to every edge of the graph in two possible ways as follows.

\paragraph{$\YES$ distribution.}
Let us choose vertex labels $\xi\in\Ftwo^N$ uniformly and independently at random, and set $b_{uv}=\xi_u\oplus\xi_v$ 
for every edge $uv$.

\paragraph{$\NO$ distribution.}
Choose every edge sign independently and uniformly from $\Ftwo$.

\vspace{10pt}

From the construction, it is clear that the labeling $x=\xi$ satisfies every constraint in every $\YES$ instance. We will later prove that every labeling violates a constant fraction of the constraints with probability $1-2^{-\Omega(N)}$ in the instances in the $\NO$ distribution (see \Cref{lem:random-sign-farness}).
Note that the $\YES$ and $\NO$ distributions do not depend on distinct values of $\ell$, unlike the original hard distributions in \cite{ambainis2011quantum}.

\paragraph{Setting up parameters}
To prove that the graph sampled from the $\NO$ distribution is far from being balanced, we need to instantiate the construction with parameters. Below, $M_0$ is a parameter that will later be used to define the range of $M$, the total number of vertices considered at the beginning of the construction, but for now we will just take $M=M_0$. As before, $\ell$ denotes the number of blocks, $N$ and $m$ denote the total number of vertices and edges in the visible graph.

We set $M_0=4\left\lceil\frac{11N}{40}\right\rceil$, and
we take $N$ to be sufficiently large so that $M_0$ is divisible by four and $\frac{11}{10}N\leq M_0<\frac{11}{10}N+4$.
At $\ell=1$, there is no sampling failure because all $N$ visible vertices fit in the single block.

Note that the ambient union of the three perfect matchings has $3M_0/2$  edges and the original graph has $M_0-N$ unselected ambient vertices, which are removed during the construction. Now removing the $M_0-N$ unused ambient vertices deletes at most $3(M_0-N)$  edges. Hence the final visible signed graph has $m$ edges where
\begin{equation}\label{eq:edge-count}
 m\geq\frac{3M_0}{2}-3(M_0-N)
 =3N-\frac{3M_0}{2}
 \geq1.3N
\end{equation}
for every sufficiently large $N$.
\cref{eq:edge-count} will be used to prove the farness of the graph sampled from the $\NO$ distribution.

We will use the notion of frustration index \cref{eq:frustration} to denote how far a signed graph is from being balanced. We now show that every graph sampled from the $\NO$ distribution described above is far from being balanced with high probability.

\begin{lemma}%
\label{lem:random-sign-farness}
Let $H=(V,E)$ be any colored multigraph~\footnote{In the hard distribution, every color class is a matching, so parallel edges can occur only between different colors. However, the final bipartiteness instance is still a simple graph as we will replace different colored constraints with paths with different color-specific internal vertices, so no parallel edges will remain.} with $\lvert V\rvert=N$ and $m$ edges. Suppose we assign every edge an independent uniform sign $b_e\in\Ftwo$ as in the construction of $\NO$ distribution. Then, for every $\tau\in(0,1/2)$,
\begin{equation}\label{eq:general-random-sign-bound}
 \Pr\!\left[\fr(H,b)\leq \tau m\right]
 \leq 2^{N-(1-h_2(\tau))m}, 
\end{equation}
where $h_2$ is the binary entropy function.
In particular, if $m\geq1.3N$, then
\[
 \Pr\!\left[\fr(H,b)\leq\frac m{50}\right]
 \leq 2^{-\Omega(N)}.
\]
\end{lemma}

\begin{proof}

Let us fix a vertex labeling $x\in\Ftwo^N$. Note that there are $2^N$ possible vertex labelings. First we will bound the bad event for one fixed vertex labeling, then take union bound over all $2^N$ labelings. Since the edge signs in the $\NO$ distribution are  uniform and independent of each other, the number of constraints violated by $x$ is distributed according to $\Bin(m,1/2)$. Therefore
\[
 \Pr\!\left[\text{$x$ violates at most $m/50$ constraints}\right]
 =2^{-m}\sum_{j\leq m/50}\binom mj
 \leq2^{-(1-h_2(1/50))m},
\]
where $h_2$ is the binary entropy function.
Taking a union bound over the $2^N$ labelings, we have the following:
\[
 \Pr\!\left[\fr(H,b)\leq\frac m{50}\right]
 \leq2^{N-(1-h_2(1/50))m}.
\]
Since $h_2(1/50)<0.142$ and $m \geq 1.3N$, the exponent above is $-\Omega(N)$. This completes the proof.
\end{proof}

As a consequence, under the $\NO$ distribution at $(M,\ell)=(M_0,1)$, $\fr(H,b)\geq0.026N$
with probability $1-2^{-\Omega(N)}$.
Thus, at $(M_0,1)$, the $\YES$ distribution is supported on balanced signed graphs, while the $\NO$ distribution has frustration index at least $0.026N$ except with exponentially small probability.

\begin{remark}\label{rem:hard-point-interpolation}
The point $(M_0,1)$ is the only parameter value at which the signed-balance tester correctness is used. As we saw, at this point, the $\YES$ distribution is supported on balanced signed graphs and the $\NO$ distribution is at constant distance from  being balanced with high probability. Consequently, correctness of the tester gives a constant lower bound on the difference of the two acceptance probabilities, that we will later denote $\Delta(M_0,1)$.

The distributions with other parameter values are auxiliary. In particular, at $(M_0,2)$ we use only the universal bound $\abs{\Delta(M_0,2)}\leq1$. More generally, the family indexed by $(M,\ell)$ is introduced so that the averaged acceptance gap becomes a structured rational function on a large grid. Since the values at $(M_0,1)$ and $(M_0,2)$ are separated by a constant, boundedness on the full grid allows the bivariate degree lemma to force a large polynomial degree.

\end{remark}

We will also need the following lemma about the probabilities of different label assignments on graph patterns (corresponding to monomials of the acceptance polynomial) for the $\YES$ and $\NO$ distributions.

\begin{lemma}%
\label{lem:forest-agreement}
Consider a colored multigraph $F$ with $v$ vertices, $e$ edges and $c$ connected components. Suppose we fix the labels on its edges according to the construction (either $\YES$ or $\NO$ case).
Under the $\NO$ distribution, the label assignment has probability $2^{-e}$. Under the $\YES$ distribution, it has probability
\[
 \begin{cases}
 2^{-(v-c)},&\text{if the signed pattern is balanced},\\
 0,&\text{otherwise}.
 \end{cases}
\]
In particular, if $F$ is a forest, the two probabilities are equal for every label assignment.
\end{lemma}

\begin{proof}
Note that under the $\NO$ distribution, since the $e$ edge labels are uniform and independent of each other, the claim follows. Under the $\YES$ distribution, due to the label assignment procedure, the assigned labels follow equations $\xi_u\oplus\xi_v=b_{uv}$
on the incident vertices. The corresponding system of equations is either inconsistent or has exactly $2^c$ solutions. This is due to the fact that once we choose one value of any root node in each connected component, all other values in that connected component are determined. Dividing by the $2^v$ possible assignments to the incident vertices gives the above mentioned probability. Note that if $F$ is a forest, every label assignment is consistent by the above argument, and $e=v-c$, so both probabilities are $2^{-e}$. This completes the proof.    
\end{proof}

\subsection{Analyzing the acceptance probability}

Consider an arbitrary quantum algorithm making $q$ queries. Let $p_{\YES}(M,\ell)$, and $p_{\NO}(M,\ell)$ be its expected acceptance probabilities under the feasible distributions (assuming no overflow event happens) under $\YES$ and $\NO$ distributions, respectively.
Now let us define the acceptance gap
\begin{equation}\label{eq:Delta}
\Delta(M,\ell)=p_{\YES}(M,\ell)-p_{\NO}(M,\ell).
\end{equation}

This case corresponds to the actual graph distribution. However, the formal acceptance probability gap $\Delta_{\mathrm{form}}(M,\ell)$ (where an overflow event is allowed) is easier to calculate algebraically, and later we will show that they are close to each other. Hence, we start by analyzing the formal version that we define below.

Let us formalize the notion of a consistent colored pattern as follows, which will be used in our proofs.

\begin{definition}[Consistent colored pattern]
A colored pattern is \emph{consistent} if, for every color, its edges form a matching and all specified oracle responses and signs are mutually consistent. Equivalently, it has nonzero probability in the distribution $P_{M,\ell}$.   
\end{definition}

\begin{definition}[Consistent graph pattern]\label{def:topology-moment}
Let $F$ be a consistent colored pattern over the vertices of the final graph as described in \cref{sec:newconstruction}. For every colored edge $e=\{u,v\}$ of color $j$ in $F$, let $Y_e$ be the indicator that the color-$j$ matching joins the ambient positions occupied by $u$ and $v$. Now consider
\begin{equation}\label{eq:topology-moment}
 Y_F:=\prod_{e\in E(F)}Y_e,
 \qquad
 W_F(M,\ell):=\mathbb E_{M,\ell}^{\mathrm{form}}[Y_F].
\end{equation}
Equivalently, $W_F(M,\ell)$ is the probability in the original hard distribution (without excluding the overflow event\footnote{If an overflow event happens, for simplicity, we assume that an arbitrary bounded degree graph with maximum degree at most $c$ is returned by the process.}) that every colored edge prescribed by $F$ occurs on the specified visible vertices. Edge signs are ignored and unspecified edges are unrestricted. In general, $W_F$ is a rational function rather than a polynomial.
\end{definition}

Note that $W_F$ is very similar to the expected value of the acceptance probability monomials (that appear in $p_\YES$ and $p_\NO$), but here the overflow event is not excluded.

The following statement is a direct consequence of the proof of Proposition~6 in~\cite{ambainis2011quantum}. Here we only include a short version of the proof that assumes familiarity with the proof in \cite{ambainis2011quantum}, but for completeness, a full proof can be found in \Cref{app:proof_3.8}.
The lemma states the full cycle-rank exponent, but the later argument uses only the consequence that every nonforest pattern contributes at least one factor of $\ell$.

\begin{restatable}%
{lemma}{lemtopologypower}\label{lem:topology-power}
Let $F$ be a consistent colored pattern with $v$ vertices, $e$ edges, and $c$ connected components, and let $W_F(M,\ell)$ be the graph pattern probability from \Cref{def:topology-moment}. There is a polynomial $B_F(M,\ell)$ of degree $O(q\log q)$ such that
\begin{equation}\label{eq:topology-factorization}
 W_F(M,\ell)
 =\frac{\ell^{e-v+c}B_F(M,\ell)}{D_q(M,\ell)}
 =\frac{\ell^{\beta(F)}B_F(M,\ell)}{D_q(M,\ell)},
\end{equation}
where $\beta(F)$ denotes the cycle-rank of $F$.    
\end{restatable}

\begin{proof}[Proof sketch]
Let us denote the connected components of $F$ as $C_1,\ldots,C_c$. The proof of~\cite[Lemma~3]{ambainis2011quantum} sums over partitions $L$ of $[c]$, where a class of the partition contains components placed in the same hidden block. In the notation of their equation~(6), the contribution associated with $L$ has the prefactor
\begin{equation}\label{eq:acl-prefactor}
 \ell^{-v+|L|}f'_L(M,\ell).
\end{equation}

Every prescribed matching edge contributes one explicit factor of $\ell$. Indeed, prescribing $r$ disjoint edges in one perfect matching contributes
\[
 R_r=\prod_{a=1}^{r}\frac{\ell}{M-(2a-1)\ell}.
\]
Since $F$ has $e$ distinct colored edges, every term entering $f'_L$ contains the factor $\ell^e$.

In the proof of \cite[Proposition~6]{ambainis2011quantum},  these explicit edge factors are removed. For the resulting expression, their partition identity shows that the numerator of the redefined $f'_L$ is divisible by $\ell^{c-|L|}$. Restoring the edge factors, the numerator of the original $f'_L$ is therefore divisible by $ \ell^{e+c-|L|}$.
Multiplying by the prefactor in \eqref{eq:acl-prefactor} leaves the power $\ell^{-v+|L|}\ell^{e+c-|L|}
 =\ell^{e-v+c}$.
Every remaining denominator factor has the form $M-(2a-1)\ell$ and is not divisible by $\ell$ as a polynomial in $M$ and $\ell$. Passing to the common denominator $D_q$ and summing over partitions therefore preserves the divisibility. The degree estimate is the one proved in~\cite[Lemma~3]{ambainis2011quantum}.
\end{proof}

\begin{remark}
\cite{ambainis2011quantum} considers only the number of $\ell$-factors needed to cancel the negative powers of $\ell$ in their expansion application. Their proof first removes the full factor $\ell^e$ and then proves the additional factor $\ell^{c-|L|}$. Keeping the full edge factor gives the cycle-rank exponent above.
\end{remark}

Now we will write the formal acceptance probability gap $\Delta^{\mathrm{form}}$. For this, we expand the degree-$2q$ acceptance polynomial from \Cref{lem:beals} into canonical monomials. Conditioned on a colored graph pattern $F$, the only difference between the $\YES$ and the $\NO$ cases is the probability of the edge signs. If $F$ is a forest, this difference is zero by \Cref{lem:forest-agreement}. If $F$ is not a forest, then $\beta(F)\geq1$, and \Cref{lem:topology-power} gives a factor of $\ell$ when writing the expression over the common denominator.

\begin{restatable}%
{proposition}{propformalgap}\label{prop:formal-gap}
There is a polynomial $H(M,\ell)$ of degree $O(q\log q)$ such that the difference between the formal acceptance probabilities in the $\YES$ and $\NO$ cases can be written as
\begin{equation}\label{eq:formal-gap}
 \Delta^{\mathrm{form}}(M,\ell)
 =\frac{\ell H(M,\ell)}{D_q(M,\ell)}.
\end{equation}    
\end{restatable}

We will formally prove this in \Cref{app:proof_3.10} and only give a proof sketch below.

\begin{proof}[Proof sketch]

Let us apply \Cref{lem:acl-rational} term by term to the acceptance polynomial and use the common denominator polynomial $D_q$. Note that the forest monomials have zero probability difference between the $\YES$ and $\NO$ cases because of \Cref{lem:forest-agreement}. Hence, every remaining monomial has positive cycle rank, so its numerator is divisible by $\ell$ by \Cref{lem:topology-power}. A linear combination of monomials divisible by $\ell$ remains divisible by $\ell$. Factoring out one  $\ell$ factor gives us \eqref{eq:formal-gap}. The degree bound follows from \Cref{lem:acl-rational}.
\end{proof}

\subsection{Bounding the sampling failure probability}

Recall that in the hard distribution described in \Cref{sec:hard_instance_acl}, it is possible that more than $M/\ell$ vertices will be assigned to a single block during the construction of the base graph. However, we now show that this event, called an \emph{overflow}, is unlikely.

The family of hard distributions used in the polynomial argument is evaluated on the grid
\begin{equation}\label{eq:grid}
 \mathcal G_N=
 \set{(M,\ell):
 1\leq\ell\leq R,
 \ \abs{M-M_0}\leq R^{3/2},
 \ 2\ell\mid M} 
 \quad \mathrm{with}\quad 
 R=\left\lfloor\frac{N^{2/3}}{100}\right\rfloor.
\end{equation}
For sufficiently large $N$, every $M$ in this grid is at least $1.09N$. %

\begin{lemma}
\label{lem:overflow}
Let us consider the construction of hard distributions as described in  \Cref{sec:hard_instance_acl} and \cref{sec:newconstruction}. Then for every $(M,\ell)\in\mathcal G_N$,
\[
 \Pr [\text{some block overflows}]
 \leq \exp(-{\Omega}(N^{1/3})).
\]
\end{lemma}

\begin{proof}
Let us fix one block, w.l.o.g. let us assume it is block $B_1$. Suppose $X$ denotes the random variable counting the number of visible vertices assigned to $B_1$. Note that $X$ is distributed according to $\Bin(N,1/\ell)$ with mean $\mu=N/\ell$. Recall that the maximum number of vertices that can be assigned to $B_1$ (i.e., its capacity) is $M/\ell$.

On $\mathcal{G}_N$, $\frac{M}{N} \geq 1.09$, so the capacity of $B_1$ is at least $(1+\rho)\mu$ for some absolute constant $\rho>0$. Using the multiplicative Chernoff bound, we have
\[
 \Pr [X>M/\ell]
 \leq \exp(-cN/\ell)
\]
for an absolute constant $c>0$. Taking a union bound over the $\ell$ blocks and using $\ell\leq R$, we obtain that
\[
 \Pr [\text{some block overflows}]
 \leq R\exp(-cN/R)
 \leq \exp(-\Omega(N^{1/3})).
\]
This completes the proof.
\end{proof}

Note that the overflow bound in \Cref{lem:overflow} is independent of the algorithm. We now combine it with the coefficient bound for the degree-$2q$ acceptance polynomial.

\begin{restatable}%
{lemma}{lemactualformal}\label{lem:actual-formal}
There is an absolute constant $C_1>0$ such that, uniformly on $\mathcal G_N$ and for $\sigma\in\{\YES,\NO\}$,
\[
 |p_{\sigma}^{\mathrm{feas}}(M,\ell)-p_{\sigma}^{\mathrm{form}}(M,\ell)|
 \leq \exp(C_1q\log N-c_1N^{1/3})
\]
whenever $q=O(N^{1/3}/\log N)$. Consequently,
\begin{equation}\label{eq:eta-bound}
 \Delta(M,\ell)
 =\Delta^{\mathrm{form}}(M,\ell)+\eta(M,\ell),
 \qquad
 \abs{\eta(M,\ell)}
 \leq 2\exp(C_1q\log N-c_1N^{1/3}).
\end{equation}
In particular, if $q\leq c_0N^{1/3}/\log N$ for a sufficiently small constant $c_0$, then $\eta=o(1)$ holds for every parameter pair in $\mathcal G_N$.    
\end{restatable}

The proof of the above lemma is deferred to \cref{app:actual-formal}.

Combining \Cref{prop:formal-gap,lem:actual-formal} gives the following statement about the feasible acceptance gap.

\begin{proposition}
\label{prop:gap-form}
There is a polynomial $H(M,\ell)$ of degree $O(q\log q)$ such that, uniformly on $\mathcal G_N$,
\begin{equation}\label{eq:gap-form}
 \Delta(M,\ell)
 =\frac{\ell H(M,\ell)}{D_q(M,\ell)}+\eta(M,\ell),
\end{equation}
where $\eta=o(1)$ when $q\leq c_0N^{1/3}/\log N$ and $c_0$ is a sufficiently small constant.
\end{proposition}

\subsection{Polynomial approximation of the denominator}

In their lower bound for expansion testing, \cite{ambainis2011quantum}  bounded the denominator polynomial $D_q(M, \ell)$ by a constant. A step for improving their result is to approximate its reciprocal by a polynomial of suitable degree.

\begin{restatable}%
{lemma}{lemreciprocal}\label{lem:reciprocal}
Let $z_1(M,\ell),\ldots,z_d(M,\ell)$ be affine functions of real variables and suppose that $\abs{z_i}\leq\frac14$ for every $i \in [d]$ throughout a region $\mathcal R$. Assume that we can write
\[
 D_q(M,\ell)=C\prod_{i=1}^{d}(1-z_i(M,\ell)),
 \qquad C>0.
\]
Then there is a polynomial $A_q(M,\ell)$ of degree at most $4d+10$ such that $\abs{D_q(M,\ell) A_q(M,\ell)-1}\leq10^{-2}$
throughout $\mathcal R$.    
\end{restatable}

We defer the formal proof of the above lemma to \Cref{app:proof_3.14}, but we include some ideas of the proof here to be able to see the explicit polynomial $A_q$.
\paragraph{Making polynomial $A_q$ explicit}

To make the reciprocal polynomial explicit, let us define 
the index set
\begin{equation}\label{eq:index-set-IK}
 \mathcal I_q
 :=\set{(a,r):1\leq a\leq 2q,\ 1\leq r\leq\lfloor 2q/a\rfloor}.
\end{equation}
Then $\abs{\mathcal I_q}=d_q$, like in \cref{eq:den-degree}. For $i=(a,r)\in\mathcal I_q$, consider the affine function
\begin{equation}\label{eq:z-factor}
 z_i(M,\ell)
 :=\frac{M_0-M+(2a-1)\ell}{M_0}.
\end{equation}
The value of $z_i$ depends only on $a$; the second coordinate $r$ only records multiplicity. 
With this notation, we can write
\begin{equation}\label{eq:denominator-z-product}
 D_q(M,\ell)
 =M_0^{d_q}\prod_{i\in \mathcal I_q}(1-z_i(M,\ell)).
\end{equation}
For a multi-index $\kappa=(\kappa_i)_{i\in\mathcal I_q}\in\mathbb Z_{\geq0}^{\mathcal I_q}$, write $|\kappa|:=\sum_{i\in\mathcal I_q}\kappa_i$.
For $j\geq0$, let
\begin{equation}\label{eq:SKj}
 S_{q,j}(M,\ell)
 :=\sum_{\substack{\kappa\in\mathbb Z_{\geq0}^{\mathcal I_q}\\ |\kappa|=j}}
       \prod_{i\in\mathcal I_q}z_i(M,\ell)^{\kappa_i}.
\end{equation}
Thus $S_{q,j}$ is the degree-$j$ complete homogeneous polynomial in the repeated affine factors $z_i$. 

Then 
$$\frac{1}{D_q(M, \ell)}
 :=
 \frac{1}{M_0^{d_q}}\sum_{j=0}^{\infty}S_{q,j}(M, \ell).$$
Let us define $A_q$ as its truncation as follows:
\begin{equation}\label{eq:AK-definition}
 J_q:=4d_q+10,
 \qquad
 A_q(M,\ell)
 :=M_0^{-d_q}\sum_{j=0}^{J_q}S_{q,j}(M,\ell).
\end{equation}
This is the total-degree-$J_q$ truncation of the formal product expansion of $D_q^{-1}$.  Note that $A_q$ is defined for every $(M, \ell)$, independently of convergence
of the infinite series.

\begin{restatable}%
{lemma}{lemdenominatorapprox}\label{lem:denominator-approx}
Suppose $q= o\left(N^{1/3}\right)$. Then the polynomial $A_q(M,\ell)$ from \eqref{eq:AK-definition} satisfies
\begin{equation}\label{eq:DA-close}
 \deg A_q\leq J_q=O(q\log q)
 \quad\text{and}\quad
 \abs{D_q(M,\ell)A_q(M,\ell)-1}\leq10^{-2}
\end{equation}
whenever $1\leq\ell\leq R$ and $\abs{M-M_0}\leq R^{3/2}$.
\end{restatable}
The proof of the above lemma is deferred to \Cref{app:proof_3.15}.

\subsection{Combining everything}

\balancedtest*

\begin{proof}
Suppose, toward a contradiction, that a tester uses
\begin{equation}\label{eq:q-assumption}
q\leq c_0\frac{N^{1/3}}{\log N}
\end{equation}
queries, where $c_0>0$ is a sufficiently small constant. Let us amplify the error probability of the tester to at most $1/20$ using a constant number of repetitions. By slight abuse of notation, we will denote $q$ for the amplified query complexity as well.

Note that at $(M_0,1)$, every $\YES$ instance is balanced. By \Cref{lem:random-sign-farness}, all except an exponentially small fraction of $\NO$ instances have frustration index at least $m/50\geq0.026N$. Hence
\begin{equation}\label{eq:Delta-one}
 \Delta(M_0,1)\geq0.9-o(1).
\end{equation}
For every parameter pair, whether or not it satisfies a testing promise,
\begin{equation}\label{eq:Delta-trivial}
 \abs{\Delta(M,\ell)}\leq1
\end{equation}
since $\Delta$ is a difference of probabilities.

Let $H$ be the polynomial from \Cref{prop:gap-form} and let $A_q$ be the polynomial from \Cref{lem:denominator-approx}. Now let us define
\begin{equation}\label{eq:P-def}
 P(M,\ell)=\frac12 H(M,\ell)A_q(M,\ell).
\end{equation}
This is a polynomial of degree $O(q\log q)$. From \cref{eq:gap-form},
\begin{equation}\label{eq:P-relation}
 P(M,\ell)
 =\frac{\Delta(M,\ell)-\eta(M,\ell)}{2\ell}
 \bigl(D_q(M,\ell)A_q(M,\ell)\bigr).
\end{equation}

On the grid $\mathcal G_N$, \Cref{lem:actual-formal} gives $\eta=o(1)$ and \Cref{lem:denominator-approx} gives
\[
 0.99\leq D_qA_q\leq1.01.
\]
Using \cref{eq:Delta-trivial} in \cref{eq:P-relation}, we obtain
\begin{equation}\label{eq:P-bounded}
 \abs{P(M,\ell)}
 \leq\frac{1.01}{2\ell}(1+o(1))<1
\end{equation}
for every point of $\mathcal G_N$ and all sufficiently large $N$.

At $(M_0,1)$, equations \eqref{eq:Delta-one} and \eqref{eq:P-relation} give
 $P(M_0,1)>0.44$
for sufficiently large $N$. At $(M_0,2)$, no correctness promise is used, only the fact that $\Delta$ is bounded by 1: equations \eqref{eq:Delta-trivial} and \eqref{eq:P-relation} give
$ \abs{P(M_0,2)}<0.26$.
Consequently,
\begin{equation}\label{eq:P-gap}
 \abs{P(M_0,1)-P(M_0,2)}>0.18.
\end{equation}

The grid has the additional divisibility condition $2\ell\mid M$. Consider $Q(X,\ell)=P(2X,\ell)$ and let $C=M_0/2$.
Then $Q$ has degree at most $\deg P$. If $\ell\mid X$ and $ \abs{X-C}\leq\frac{R^{3/2}}2$,
then $(2X,\ell)\in\mathcal G_N$, so \eqref{eq:P-bounded} gives $\abs{Q(X,\ell)}\leq1$. The gap \eqref{eq:P-gap} holds at $(C,1)$ and $(C,2)$. Applying \Cref{lem:acl-degree} yields
\[
 \deg P\geq\deg Q=\Omega(\sqrt R)=\Omega(N^{1/3}).
\]
On the other hand, $\deg P=O(q\log q)$, 
thus $q\log q=\Omega(N^{1/3})$.
Since $q\leq N$, we get $
 q=\Omega\!\left(\frac{N^{1/3}}{\log N}\right)$ 
contradicting \eqref{eq:q-assumption} when $c_0$ is sufficiently small.
\end{proof}

\subsection{Reduction to bipartiteness testing}

We now reduce the signed graph instance into a simple graph such that when the signed graph instance is balanced, the constructed simple graph is bipartite and when the signed graph has large frustration index, the resulting graph is far from being bipartite. This is a standard reduction, similar to \cite{DBLP:conf/www/AdriaensA23,Zaslavsky18}.

Let $(H,b)$ be the given signed colored graph instance.
Now for each vertex $u\in[N]$ in $H$ and each color $j\in\{1,2,3\}$, let us introduce two auxiliary vertices $a_{u,j}, p_{u,j}$.
For a colored signed edge $uv$ of color $j$, let us orient its description so that $u<v$. Now we do the following:
\begin{itemize}
\item If $b_{uv}=0$, insert the even path $u-a_{u,j}-v$.

\item If $b_{uv}=1$, let us insert the odd path $u-a_{u,j}-p_{u,j}-v$.

\end{itemize}
All unused auxiliary vertices are isolated. Let $G(H,b)$ be the graph constructed as above.

Note that the graph $G(H,b)$ has $7N$ vertices where $|V(H)|=N$. Since a vertex is incident to at most one edge of each color, every auxiliary vertex belongs to at most one path and  differently colored parallel edges of $H$ connect distinct vertices in $G$. Thus $G(H,b)$ is simple and has maximum degree three.

Now we have the following lemma that shows the reduction is valid.

\begin{lemma}\label{lem:reduction}
For every signed colored multigraph $(H,b)$ produced above, the graph $G(H,b)$ satisfies
\begin{equation}\label{eq:balance-bip}
 G(H,b)\text{ is bipartite}
 \quad\Longleftrightarrow\quad
 (H,b)\text{ is balanced},
\end{equation}
and
\begin{equation}\label{eq:distance-preservation}
 \min\{\left|\text{edges deleted from $G(H,b)$ to make it bipartite}\right|\}
 =\fr(H,b).
\end{equation}
Moreover, each adjacency-list query to $G(H,b)$ can be simulated using a constant number of queries to the signed graph $H$.
\end{lemma}

\begin{proof}
Following the construction described above, note that in a bipartite graph, an even path has endpoints of the same side of a correct bipartition, while an odd path has endpoints of opposite sides. Therefore a correct bipartition of the original vertices is equivalent to a labeling satisfying \cref{eq:balance}. This proves \cref{eq:balance-bip}.

The paths obtained after introducing the auxiliary vertices corresponding to distinct colored edges are edge-disjoint. Let us fix a labeling of the original vertices. If a signed constraint is violated by this labeling, at least one edge must be removed from the corresponding path before the path can be properly labeled with the prescribed endpoint signs. Conversely, let us delete one edge from every path corresponding to a violated constraint. Every remaining path component can then be labeled consistently with the endpoint signs. Minimizing over all labelings gives us \cref{eq:distance-preservation}. Since adding new edges can not change a far from bipartite graph to a bipartite graph, the required minimum number of edge modifications remains the same.

For the query simulation, note that if the algorithm queries a neighbor of such a vertex of $G$ that is an original vertex in $H$, the simulation algorithm queries its at most three neighbors in $H$, determines the corresponding auxiliary vertex according to the returned endpoint and sign, and returns the requested neighbor. A query to $a_{u,j}$ or $p_{u,j}$ in $G$ is answered by querying $(u,j)$ in $H$ and checking the returned endpoint and sign. This computation is reversible, and can be performed coherently in the standard way.
\end{proof}

\theobiplb*

\begin{proof}
Note that, from \Cref{lem:random-sign-farness}, we know that with probability $1-2^{-\Omega(N)}$ a $\NO$ signed instance at $(M_0,1)$ satisfies $\fr(H,b)\geq\frac m{50}\geq0.026N$.
By \Cref{lem:reduction}, the corresponding ordinary graph requires at least $0.026N$ edge deletions to become bipartite. It has $7N$ vertices and maximum degree three, so it is more than $\frac{0.026N}{3\cdot7N}>10^{-3}$-far from bipartite.

Suppose a $7N$-vertex bipartiteness tester uses $q$ queries. The simulation argument in \Cref{lem:reduction} gives a  tester for balancedness testing using $O(q)$ queries to the signed graph. \Cref{theo:balanced} therefore gives $q=\Omega\!\left(\frac{N^{1/3}}{\log N}\right)$.
For graphs whose number of vertices are not of the form of $7N$, we can add at most six isolated vertices to handle every sufficiently large $N$. This completes the proof.
\end{proof}

\section{Quantum query lower bound for expansion testing}\label{sec:expansiontesting}

In this section, we prove that expansion property testing in the bounded-degree model requires $\Omega(N^{1/3}/\log^{4/3} N)$ quantum queries.

\begin{restatable}[Formal version of \cref{thm:informal-expansion}]{theorem}{theoexpansionlb}\label{theo:expansionlb}
Let $G$ be a bounded degree graph on $N$ vertices and $\eps, \alpha, \alpha' \in (0,1)$ be constants with $0< \alpha' \le \alpha <1$. In order to distinguish with probability at least $2/3$ if $G$ is an $\alpha$-expander, or $\eps$-far from being $\alpha'$-expander, $\Omega(N^{1/3}/(\log N)^{4/3})$ quantum queries are necessary.    
\end{restatable}

For constructing the hard instances in bipartiteness testing, we considered graphs with $3$ random matchings. However, this will not be enough for expansion testing since any visible vertex in the final graph has a constant probability of losing all of its matching neighbors.

To bypass this bottleneck, we consider graphs with $\Theta(\log N)$ matchings. This will ensure that every cut in the graph will have sufficient expansion with high probability even when a constant fraction of the ambient vertices is unused.
Finally, we will apply a replacement-graph gadget to reduce the degree of every vertex to four.

\subsection{Description of our hard instances}\label{sec:expansioninstance}

We consider the hard distribution $P^{(c)}_{M,\ell}$ in \cref{sec:hard_instance_acl} with the number of matchings $c=\Theta(\log N)$ an even number. Remember that the colors of the matchings are indexed by $[c]$. Similarly to \cite{ambainis2011quantum}, we obtain $\YES$ instances with high probability by setting the construction's block number $\ell=1$. For $\ell=2$, we can only obtain $\NO$ instances (i.e. far from expander graphs) because the resulting graph will be disconnected.

This way, we first construct an intermediate graph with degree $\Theta(\log N)$, then we will use a replacement graph gadget to transform this intermediate graph to a final graph of maximum degree $4$. To show that the final graphs that we obtain from intermediate graphs in the $\YES$ distribution $P^{(c)}_{M,1}$ are expanders, we will first prove that the intermediate graph is an expander, and then prove that the final graph remains an expander as well. Since every graph in the support of the $\NO$ distribution $P^{(c)}_{M,2}$ has two large sets of vertices with no edges between them, it is easy to see that a similar statement holds in this case as well.

\paragraph{Setting up parameters}
Let us set $c=\Theta(\log N)$, and without loss of generality, let us assume that $c$ is even. Let $P^{(c),\mathrm{feas}}_{M,\ell}$ denote the feasible distribution with $c$ independent uniformly random perfect matchings in every hidden block conditioned on no overflow event happening. %

We set up the same grid as we used for the bipartiteness testing lower bound.
\begin{equation}\label{eq:replacement-grid}
 \mathcal G_N=
 \set{(M,\ell):1\leq\ell\leq R,
       \abs{M-M_0}\leq R^{3/2},
       2\ell\mid M}
       \text{ with }
       M_0=4\left\lceil\frac{11N}{40}\right\rceil
       ,
       R=\left\lfloor\frac{N^{2/3}}{100}\right\rfloor.
\end{equation}

Let us start by proving that the intermediate graph is an expander.

\begin{restatable}%
{lemma}{lemlogdegreeexpansion}\label{lem:logdegree-expansion}
Let $H$ be drawn from $P^{(c), \mathrm{feas}}_{M_0,1}$. With probability $1-o(1)$,
\begin{equation}\label{eq:intermediate-expansion}
 \left|e_H(U,V(H)\setminus U)\right| \geq \frac{c}{20}\abs{U}
\end{equation}
holds for every $U\subseteq V(H)$ with $\abs{U}\leq N/2$.    
\end{restatable}

The proof of the above lemma is in \cref{app:proof_4.2}.

{

Now we show that the intermediate $\NO$ instances are far from being expander.

\begin{lemma}
\label{lem:intermediate-no-farness}
Fix $\alpha'>0$. For sufficiently large $N$, every graph $H$
in the support of $P^{(c),\mathrm{feas}}_{M_0,2}$ contains a set
$U\subseteq[N]$ such that
\begin{equation}\label{eq:intermediate-no-cut}
 \frac N3\leq\abs U\leq\frac N2,
 \qquad
 \partial_H(U)=\varnothing.
\end{equation}
Thus $H$ is $\alpha'/3c$-far from every $\alpha'$-vertex-expander.
\end{lemma}

\begin{proof}
First let us suppose that no block overflows, and let $V_1,V_2$ be the two
visible vertex classes assigned to the hidden blocks. Since
$\abs{V_1}+\abs{V_2}=N$ and each hidden block has $M_0/2$ ambient
positions, $\abs{V_i}\leq M_0/2$. Hence
\[
 \abs{V_i}\geq N-\frac{M_0}{2}
 \geq\frac{9N}{20}-2
 \geq\frac N3
\]
for all sufficiently large $N$. Every matching edge has both endpoints
in the same hidden block. Therefore there is no edge between $V_1$
and $V_2$. Taking $U$ to be the smaller of $V_1,V_2$ proves
\cref{eq:intermediate-no-cut}.

Now let us fix the set $U$ stated above. In any $\alpha'$-vertex-expander,
$U$ must have at least $\alpha'\abs U$ distinct outside neighbors.
Initially it has no outside neighbor. Each added edge creates at most one new outside
neighbor of $U$, while edge deletions create none. Thus at least
$\alpha'\abs U\geq\alpha'N/3$ edge modifications are necessary.
Dividing by the bounded-degree normalization $cN$ finishes the proof.
\end{proof}

}

\subsection{The replacement-product construction}

We use the following standard consequence of Friedman's second-eigenvalue
theorem and the discrete Cheeger inequality; see
\cite[Theorems~7.10 and~4.11]{HLW06}.

\begin{lemma}
\label{lem:3regular-expanders}
There exists an absolute constant $h_0>0$, such that for every sufficiently large even integer $n$, there is a simple $3$-regular graph $J_n$ on vertex set $[n]$ such that
\begin{equation}\label{eq:cloud-expansion}
 |e_{J_n}(A,[n]\setminus A)|
 \geq h_0\min\set{\abs A,n-\abs A}
\end{equation}
for every $A\subseteq[n]$.
\end{lemma}

\paragraph{Construction:}
For sufficiently large $N$, let us consider a fixed graph $J_{c}$ according to
\Cref{lem:3regular-expanders} and label its vertices by the matching
colors $[c]$. This is valid since $c$ is even and is sufficiently large for large enough $N$.

Given an intermediate colored multigraph $H$, let us define its replacement graph $\mathcal R(H)$ as follows.
\begin{enumerate}
\item Replace every vertex $u\in V(H)$ by a set of $c$ vertices $C_u=\set{(u,j):j\in[c]}$, and insert a copy of $J_{c}$ on vertex set $C_u$. These edges are called internal edges.

\item For every color-$j$ edge $\{u,v\}$ of $H$, add the external edge $\{(u,j),(v,j)\}$.

\end{enumerate}

\begin{figure}[htbp]
\centering
\begin{tikzpicture}[
  font=\small,
  hvertex/.style={circle,draw=black!70,fill=white,
    minimum size=6.5mm,inner sep=0pt},
  cloudvertex/.style={circle,draw=black!65,fill=white,
    minimum size=2.8mm,inner sep=0pt},
  jport/.style={cloudvertex,draw=blue!70!black,fill=blue!18,
    line width=.7pt},
  kport/.style={cloudvertex,draw=red!70!black,fill=red!16,
    line width=.7pt}
]
\node[hvertex] (u) at (0,0) {$u$};
\node[hvertex] (v) at (1.75,0) {$v$};
\draw[blue!70!black,line width=.9pt,bend left=24]
  (u) to node[above=-1pt] {$j$} (v);
\draw[red!70!black,line width=.9pt,bend right=24]
  (u) to node[below=-1pt] {$k$} (v);
\node[font=\scriptsize,align=center] at (.875,-.92)
  {intermediate graph $H$};

\draw[-{Latex[length=2.2mm]},line width=.75pt]
  (2.35,0) -- node[above,font=\scriptsize] {replace} (3.35,0);

\coordinate (ujp) at (5.05,0);
\coordinate (ukp) at (4.64,.71);
\coordinate (uap) at (3.82,.71);
\coordinate (ubp) at (3.41,0);
\coordinate (ucp) at (3.82,-.71);
\coordinate (udp) at (4.64,-.71);
\coordinate (vxp) at (9.24,0);
\coordinate (vyp) at (8.83,.71);
\coordinate (vkp) at (8.01,.71);
\coordinate (vjp) at (7.60,0);
\coordinate (vzp) at (8.01,-.71);
\coordinate (vwp) at (8.83,-.71);

\draw[dashed,black!48] (4.23,0) circle (1.12);
\draw[dashed,black!48] (8.42,0) circle (1.12);
\draw[black!42,line width=.55pt]
  (ujp)--(ukp)--(uap)--(ubp)--(ucp)--(udp)--cycle
  (ujp)--(ubp) (ukp)--(ucp) (uap)--(udp);
\draw[black!42,line width=.55pt]
  (vxp)--(vyp)--(vkp)--(vjp)--(vzp)--(vwp)--cycle
  (vxp)--(vjp) (vyp)--(vzp) (vkp)--(vwp);

\draw[blue!70!black,line width=1pt]
  (ujp) -- node[below=2pt,font=\scriptsize,fill=white,inner sep=1pt]
  {color $j$} (vjp);
\draw[red!70!black,line width=1pt]
  (ukp) -- node[above=2pt,font=\scriptsize,fill=white,inner sep=1pt]
  {color $k$} (vkp);

\node[jport] (uj) at (ujp) {};
\node[kport] (uk) at (ukp) {};
\node[cloudvertex] at (uap) {};
\node[cloudvertex] at (ubp) {};
\node[cloudvertex] at (ucp) {};
\node[cloudvertex] at (udp) {};
\node[cloudvertex] at (vxp) {};
\node[cloudvertex] at (vyp) {};
\node[kport] (vk) at (vkp) {};
\node[jport] (vj) at (vjp) {};
\node[cloudvertex] at (vzp) {};
\node[cloudvertex] at (vwp) {};

\node[above,font=\scriptsize] at (uj.north) {$(u,j)$};
\node[above left=-1pt,font=\scriptsize] at (uk.north) {$(u,k)$};
\node[above,font=\scriptsize] at (vj.north) {$(v,j)$};
\node[above right=-1pt,font=\scriptsize] at (vk.north) {$(v,k)$};
\node[font=\scriptsize] at (4.23,1.38) {$C_u\cong J_{c}$};
\node[font=\scriptsize] at (8.42,1.38) {$C_v\cong J_{c}$};
\node[font=\scriptsize,align=center] at (6.33,-1.30)
  {replacement graph $\mathcal R(H)$};
\end{tikzpicture}
\caption{The replacement graph construction. Each vertex $u$ of the intermediate graph is replaced by a copy of $J_{c}$, whose color-$j$ vertex is denoted $(u,j)$. A color-$j$ edge $\{u,v\}$ becomes the external edge $\{(u,j),(v,j)\}$. The color-$k$ edge illustrates that differently colored parallel intermediate edges use different neighbors. Only a few neighbors and internal cloud edges are shown schematically.}
\label{fig:replacement-construction}
\end{figure}
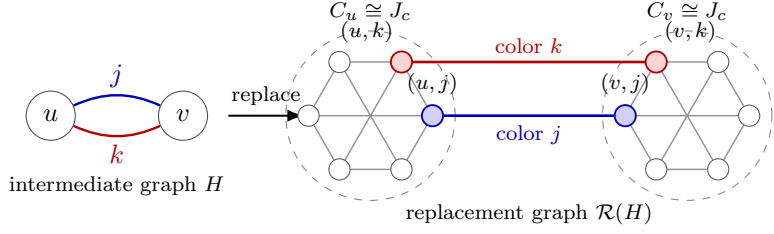

\Cref{fig:replacement-construction} illustrates this standard replacement-graph construction. It has $cN$ vertices and it is simple: different colors use different auxiliary vertices, even when two edges of different colors join the same pair of intermediate vertices. Every vertex has three internal neighbors and at most one external neighbor, so the maximum degree of this graph is four.

Now we prove that if a graph $H$ is an expander, then its corresponding replacement graph $\mathcal R(H)$ is an expander as well.

\begin{restatable}%
{lemma}{lemreplacementexpansion}\label{lem:replacement-expansion}
Suppose an $N$-vertex colored multigraph $H$ satisfies
\[
 |e_H(U,V(H)\setminus U)|
 \geq \beta c\min\set{\abs U,N-\abs U}
\]
for every $U\subseteq V(H)$. Then $\mathcal R(H)$ has vertex expansion at least $\alpha$ depending only on $\beta$ and $h_0$.    
\end{restatable}

The proof of the above lemma is deferred to \cref{app:proof_4.5}.

The above lemmas show that a $\YES$-instance graph remains an expander after the replacement graph construction. Now we prove that a $\NO$-instance graph remains far from being an expander after the replacement graph construction.

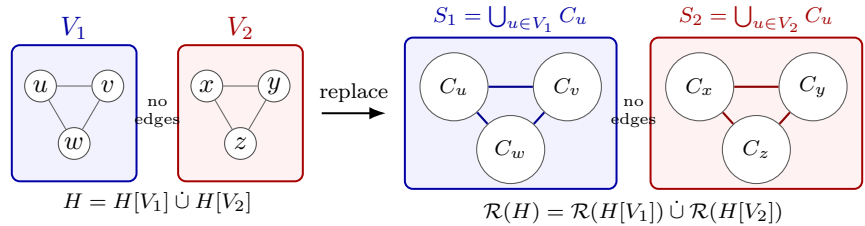
\begin{figure}[H]
\centering
\begin{tikzpicture}[
  font=\small,
  hvertex/.style={circle,draw=black!70,fill=white,
    minimum size=4.2mm,inner sep=0pt},
  cloud/.style={circle,draw=black!65,fill=white,
    minimum size=9mm,inner sep=1pt,font=\scriptsize},
  firstblock/.style={rounded corners=3pt,draw=blue!65!black,
    fill=blue!5,line width=.75pt},
  secondblock/.style={rounded corners=3pt,draw=red!65!black,
    fill=red!5,line width=.75pt}
]
\draw[firstblock] (-.10,-.90) rectangle (1.55,.90);
\draw[secondblock] (2.10,-.90) rectangle (3.75,.90);
\node[blue!65!black] at (.725,1.13) {$V_1$};
\node[red!65!black] at (2.925,1.13) {$V_2$};

\node[hvertex] (u1) at (.28,.36) {$u$};
\node[hvertex] (u2) at (1.17,.36) {$v$};
\node[hvertex] (u3) at (.725,-.40) {$w$};
\draw[black!60] (u1)--(u2)--(u3)--(u1);

\node[hvertex] (v1) at (2.48,.36) {$x$};
\node[hvertex] (v2) at (3.37,.36) {$y$};
\node[hvertex] (v3) at (2.925,-.40) {$z$};
\draw[black!60] (v1)--(v2)--(v3)--(v1);

\node[font=\tiny,align=center] at (1.825,0) {no\\edges};
\node[font=\scriptsize] at (1.825,-1.18)
  {$H=H[V_1]\mathbin{\dot\cup}H[V_2]$};

\draw[-{Latex[length=2.2mm]},line width=.75pt]
  (4.00,0) -- node[above,font=\scriptsize] {replace} (4.85,0);

\draw[firstblock] (5.10,-1.00) rectangle (7.90,1.00);
\draw[secondblock] (8.35,-1.00) rectangle (11.15,1.00);
\node[blue!65!black,font=\scriptsize] at (6.50,1.25)
  {$S_1=\bigcup_{u\in V_1}C_u$};
\node[red!65!black,font=\scriptsize] at (9.75,1.25)
  {$S_2=\bigcup_{u\in V_2}C_u$};

\node[cloud] (c1) at (5.75,.35) {$C_u$};
\node[cloud] (c2) at (7.25,.35) {$C_v$};
\node[cloud] (c3) at (6.50,-.48) {$C_w$};
\draw[blue!55!black,line width=.8pt]
  (c1)--(c2)--(c3)--(c1);

\node[cloud] (d1) at (9.00,.35) {$C_x$};
\node[cloud] (d2) at (10.50,.35) {$C_y$};
\node[cloud] (d3) at (9.75,-.48) {$C_z$};
\draw[red!55!black,line width=.8pt]
  (d1)--(d2)--(d3)--(d1);

\node[font=\tiny,align=center] at (8.125,0) {no\\edges};
\node[font=\scriptsize] at (8.125,-1.31)
  {$\mathcal R(H)=\mathcal R(H[V_1])
    \mathbin{\dot\cup}\mathcal R(H[V_2])$};
\end{tikzpicture}
\caption{A $\NO$ instance before and after the replacement construction
(in the non-overflow case). Only a few vertices,
clouds, and edges are shown schematically; blue and red distinguish the
two visible block classes, not matching colors.}
\label{fig:replacement-no-instance}
\end{figure}

As shown in \Cref{fig:replacement-no-instance}, let $V_1,V_2$ denote the  classes of visible vertices
corresponding to the two hidden blocks from the hard distribution construction. Every matching edge of the intermediate graph remains
inside its class, so
$H=H[V_1]\mathbin{\dot\cup}H[V_2]$. Replacing each vertex $v$ by cloud $C_v$ and
adding the corresponding external edges preserves this separation. Thus,
for $S_i=\bigcup_{u\in V_i}C_u$, there is no edge between $S_1$ and $S_2$
and $\partial_{\mathcal R(H)}(S_i)=\varnothing$. 

\begin{lemma}
\label{lem:replacement-no-farness}
For every constant $\alpha'>0$ and all sufficiently large $N$, let
$H$ be drawn from $P^{(c),\mathrm{feas}}_{M_0,2}$ and let
$G=\mathcal R(H)$. Every graph in the support of $G$ is $\eps$-far from
being an $\alpha'$-vertex-expander, where 
$\eps=\frac{\alpha'}{12}$.
In particular, this holds with probability one.
\end{lemma}

\begin{proof}
{
Let us denote $\bar n=Nc$ for the number of vertices of $G$. Assume that in the hard instance construction of $H$, no block overflows, and let $V_1,V_2$ be the two
visible block classes of $H$. Since $\abs{V_1}+\abs{V_2}=N$ and
$\abs{V_i}\leq M_0/2$, the definition of $M_0$ gives
\[
 \abs{V_i}\geq N-\frac{M_0}{2}
 \geq \frac{9N}{20}-2
 \geq \frac N3
\]
for all sufficiently large $N$. Every matching edge lies within one of the hidden
blocks, so there are no crossing edges between $V_1$ and $V_2$.
The smaller of the corresponding two unions of replacement clouds is a
set $S$ satisfying
\[
 \frac{\bar n}{3}\leq\abs S\leq\frac{\bar n}{2},
 \qquad
 \partial_G(S)=\varnothing.
\]

Thus, transforming $G$ into an $\alpha'$-vertex-expander
requires the fixed set $S$ to acquire at least
$\alpha'\abs S\geq\alpha'\bar n/3$ distinct outside neighbors. Each
added crossing edge creates at most one such neighbor, thus at least $\alpha'\bar n/3$ edge modifications are necessary. Dividing by $4\bar n$ because of the degree and size of the replacement graph $G=\mathcal{R}(H)$, gives
distance at least $\alpha'/12$.}
\end{proof}

Now we argue that the replacement graph problem is at least as hard as the intermediate graph problem, up to a constant factor. We prove this by a reduction argument.

\begin{lemma}
\label{lem:replacement-query-reduction}
Let $H$ be an intermediate graph as described in \cref{sec:expansioninstance} and let $\mathcal R(H)$ be
the corresponding replacement graph as described above. For every quantum algorithm $\mathcal A$ making
$q$ adjacency-list queries to $\mathcal R(H)$, there is a quantum
algorithm $\mathcal B$ making at most $2q$ queries to
$H$ such that, for every $H$,
\[
 \Pr[\mathcal B^H\text{ accepts}]
 =
 \Pr[\mathcal A^{\mathcal R(H)}\text{ accepts}].
\]
Consequently, a $q$-query algorithm distinguishing two distributions of
replacement graphs yields a $2q$-query algorithm distinguishing the
corresponding distributions of intermediate graphs with the same
advantage.
\end{lemma}

\begin{proof}
Following the construction described above, note that every vertex of
$\mathcal R(H)$ has the form $(u,j)$.  Its first three neighbors are determined by
the copy of the fixed graph $J_c$ on the cloud $C_u$, and hence they can be
returned without making any query to $H$.  The possible fourth neighbor is
determined by the color-$j$ edge incident to $u$ in $H$.  A query to $(u,j)$
in $H$ returns a vertex $v$ if such an edge $\{u,v\}$ is present and returns
$\bot$ otherwise.  Accordingly, the simulation returns $(v,j)$ in the first
case and $\bot$ in the second case.

For the query simulation, the above computation is performed coherently.
For each query made by $\mathcal A$, the algorithm $\mathcal B$ first queries
$H$ to compute the color-$j$ neighbor of $u$ in a clean work register.  It
uses this value only when the queried neighbor is the fourth neighbor, writes the
corresponding answer in the answer register, and then queries $H$ once more
to return the work register to zero.  The queries to the first three neighbors are answered
reversibly from the fixed description of $J_c$.  Thus one adjacency-list
query to $R(H)$ is simulated exactly, even when the query is in
superposition, using at most two queries to $H$ and leaving no garbage.

Replacing every oracle call of $\mathcal A$ by this simulation gives an
algorithm $\mathcal B$ making at most $2q$ queries to $H$.  Since the
simulation is exact, the final states of the two computations agree, and
hence
\[
  \Pr[{\mathcal B}^H \text{ accepts}]
  =
  \Pr[{\mathcal A}^{R(H)} \text{ accepts}].
\]
The final statement follows immediately.

\end{proof}

\begin{corollary}\label{cor:replacement-property-gap}
There is a constant $\alpha>0$ such that, for every fixed
$0<\alpha'<\alpha$, setting $\eps=\alpha'/12$ gives the following.
\begin{enumerate}
\item If $H$ is drawn from $P^{(c),\mathrm{feas}}_{M_0,1}$, then $\mathcal R(H)$ is an $\alpha$-vertex-expander with probability $1-o(1)$.
\item If $H$ is drawn from $P^{(c),\mathrm{feas}}_{M_0,2}$, then $\mathcal R(H)$ is $\eps$-far from being an $\alpha'$-vertex-expander with probability $1$.
\end{enumerate}
\end{corollary}

\begin{proof}
Applying \Cref{lem:logdegree-expansion,lem:replacement-expansion} with
$\beta=1/20$, we have the first statement. Now if we apply
\Cref{lem:replacement-no-farness} with the  value of $\alpha'$ stated in the corollary, we have the second result.
\end{proof}

\begin{remark}
Similarly to the proof of the bipartiteness testing lower bound, here we will also consider two feasible acceptance probability values $p(M_0,1)$ and $p(M_0,2)$ for the $\YES$ and $\NO$ distributions, respectively. However, unlike in the case of the bipartiteness lower bound, here every graph sampled from $P_{M_0,2}$ is far from expander. As a result, we do not consider the difference polynomial for every parameter here. This is the approach that was considered in \cite{ambainis2011quantum}.
\end{remark}

\subsection{The rational acceptance probability with logarithmically many colors}

Here we prove that the acceptance probability distribution analysis for bipartiteness testing in \Cref{sec:bip_lb} remains valid for $c=\Theta(\log N)$ colored matchings.

\begin{restatable}{lemma}{lemreplacementrational}\label{lem:replacement-rational}
Consider a quantum algorithm that makes $q$ queries. In the hard distribution denoted by $P^{(c)}_{M,\ell}$, its expected acceptance probability has the form $\frac{F(M,\ell)}{D_q(M,\ell)}$ with $\deg(F)=O(q\log q)$,
where $D_q$ is the denominator in \eqref{eq:denominator}. On the grid $\mathcal G_N$, the expected acceptance probability conditioned on no overflow event satisfies
\begin{equation}\label{eq:replacement-rational-error}
 p(M,\ell)=\frac{F(M,\ell)}{D_q(M,\ell)}+\eta(M,\ell),
 \qquad
 |\eta(M, \ell)| \leq
 \exp\!\left(O(q\log N)-\Omega(N^{1/3})\right).
\end{equation}
\end{restatable}

\begin{proof}
By \Cref{lem:beals,lem:acl-rational}, the degree-$2q$ acceptance
polynomial can be averaged monomial by monomial, and every formal
monomial expectation has denominator dividing $D_q$. The multiplicity
count in \cref{eq:denominator} is unchanged when the number of matching
colors is $c$: a degree-$2q$ monomial specifies at most $2q$ colored
matching edges, while an occurrence of $M-(2a-1)\ell$ requires at least
$a$ specified edges in one matching and one block and hence occurs at
most $\lfloor 2q/a\rfloor$ times. Thus $D_q$ remains a common
denominator, and \cref{eq:den-degree} gives
$\deg D_q=O(q\log q)$. By \Cref{lem:acl-rational}, the numerator has the same degree bound. Changing the number of available
colors from $3$ to $\Theta(\log N)$ does not change this calculation.

There are $Nc$ possible queries and $O(N)$ possible responses per query. The coefficient argument of \Cref{lem:coefficient-bound} therefore gives total coefficient absolute value
\[
 \exp(O(q\log(Nc)))=\exp(O(q\log N)).
\]
On $\mathcal G_N$, every block has capacity at least $1.09$ times its
mean occupancy. The overflow bound \Cref{lem:overflow} applies
unchanged, because it depends only on the block assignments, and gives
\[
 \Pr[\text{some block overflows}]
 \leq \exp(-{\Omega}(N^{1/3})).
\]

Similar to the proof of \Cref{lem:actual-formal}, the formal
and feasible expectations of a Boolean monomial differ only on this
event. Summing with \Cref{lem:coefficient-bound} proves
 \cref{eq:replacement-rational-error}. 
\end{proof}

\paragraph{Size-adjustment and a subtlety}
Note that the replacement graph construction described above constructs graphs with $Nc$ vertices. To prove our lower bound for  graphs on an arbitrary number of vertices, we need to add additional vertices. However, adding isolated vertices, which is a common approach to achieve such reductions, will not work here. In particular, adding isolated vertices will make the graph in the $\YES$ instance no longer an expander, a promise that is required for the lower bound. Instead, we will add additional vertices by subdividing the existing edges. For example, we will add a vertex $z$ of degree $2$ by subdividing a fixed internal cloud edge $xy$ into the path $x-z-y$. Repeating this procedure $O(\log N)$ times yields a graph with $\Theta(Nc)$ vertices.  The following two lemmas show that subdividing a sufficiently small
number of edges preserves both the $\NO$ and the $\YES$ instance promises. We first argue the $\NO$ instances.

{

\begin{restatable}{lemma}{lemnosubdivision}\label{lem:no-subdivision}
Let $\alpha'>0$ be a fixed constant. Let $H$ be any graph in the support of
$\mathcal{P}^{(c),\mathrm{feas}}_{M_0,2}$, let
$G = \mathcal{R}(H), \overline n := |V(G)| = Nc$. Suppose $G'$ is the graph obtained from $G$ by subdividing $r \le \overline n/4$ edges, each at most once. Then, for all sufficiently large $N$, $G'$ is $\alpha'/12$-far from every $\alpha'$-vertex-expander.    
\end{restatable}

The proof of this lemma is deferred to \cref{app:proofof_4.11}.

}

\begin{restatable}{lemma}{lemboundedsubdivision}\label{lem:bounded-subdivision}
Let $G$ be a graph of maximum degree at most four that, for some parameter $\rho>0$, satisfies
\[
 |e_G(U,V(G)\setminus U)|
 \geq \rho \min\set{\abs U, \abs{V(G)}-\abs U}
\]
for every $U\subseteq V(G)$. If at most $\abs{V(G)}/4$ edges are subdivided, each at most once, the resulting graph has vertex expansion at least a positive constant depending only on $\rho$.   
\end{restatable}

We defer the proof of this lemma to \cref{app:proofof_4.12}.

\subsection{Combining everything}

\theoexpansionlb*

\begin{proof}
We first prove the result for the intermediate graph on $N'=N/c$ vertices, and then by the reduction in \cref{lem:replacement-query-reduction} the main result follows.
This way, we first assume that the replacement graph has exactly $\bar n=N'c$ vertices, where $c=\Theta(\log N)$. 

Suppose a tester for the intermediate graph makes $q$ queries. 
Let us amplify the error probability of the tester to at most $1/20$ using a constant number of repetitions. By slight abuse of notation, we will denote $q$ for the amplified query complexity as well. Moreover, let $p(M,\ell)$ be its expected acceptance probability.

By \Cref{cor:replacement-property-gap}, we know that $
 p(M_0,1)\geq\frac{19}{20}-o(1)
 \text{, and }
 p(M_0,2)\leq\frac{1}{20}+o(1)$.

Assume toward a contradiction that $q\leq c_0\frac{N'^{1/3}}{\log N'}$
for a sufficiently small constant $c_0$. Now let us consider $F$ and $\eta$ as in \Cref{lem:replacement-rational}. Since the grid and denominator are the same as in the signed proof, \Cref{lem:denominator-approx} gives a polynomial $A_q(M,\ell)$ satisfying
$\deg A_q=O(q\log q)$ and $\abs{D_qA_q-1} \leq 10^{-2}$ throughout the grid. Now let us consider the polynomial:
\[
 P(M,\ell)=\frac12F(M,\ell)A_q(M,\ell).
\]
Then $\deg P=O(q\log q)$ and we have
\[
 P(M,\ell)=\frac{p(M,\ell)-\eta(M,\ell)}2
            D_q(M,\ell)A_q(M,\ell).
\]
By \Cref{lem:replacement-rational}, $\eta(M, \ell) = o(1)$ uniformly under the assumed query bound. Since $0\leq p\leq1$, the polynomial is bounded by one on $\mathcal G_{N'}$. Moreover, $P(M_0,1)>0.46$ and $\abs{P(M_0,2)}<0.04$
for sufficiently large $N'$. Rescaling $M=2X$ and applying \Cref{lem:acl-degree} gives $\deg P=\Omega(\sqrt R)=\Omega(N'^{1/3})$.
The upper bound $\deg P=O(q\log q)$ therefore implies $q=\Omega\!\left(\frac{N'^{1/3}}{\log N'}\right)=\Omega\!\left(\frac{(N/c)^{1/3}}{\log  N}\right)=\Omega\!\left(\frac{N^{1/3}}{\log^{4/3} N}\right)$, 
which is in contradiction with our initial assumption for small enough constant $c_0$.

Note that the argument above considers graphs with $N'c$ vertices. For graphs whose number of vertices are not of this form, we will apply the size-adjustment argument described before. This completes the proof.
\end{proof}

\section*{Acknowledgments}
The authors are grateful to Frédéric Magniez for initial discussions and the project idea of improving the expansion testing lower bound and exploring the possibility of adapting the method for bipartiteness testing.

CK is supported by French PEPR integrated project EPiQ 
(ANR-22-PETQ-0007) and partially supported by ANR Grant FLITTLA (ANR-21-CE48-0023).
SS's research is supported by the NRF Investigatorship award (NRF-NRFI10-2024-0006), CQT Young Researcher Career Development Grant (25-YRCDG-SS) and the grant ANR-18-IDEX-0001 between Université Paris Cité and National University of Singapore. 
DS's research is supported by the German Federal Ministry of Research, Technology and Space (QuSol, 13N17173), by the Deutsche Forschungsgemeinschaft under Germany’s Excellence Strategy – EXC-2111 – 390814868, and by the Munich Quantum Valley.

\paragraph{Statement of AI usage:}
ChatGPT-6 Astra was used extensively in the development of this paper, including in exploring and refining proof ideas and in drafting and revising the text. 
In particular, the idea of introducing an intermediate problem in the lower bound proof of bipartiteness testing, as well as the modifications of the \cite{ambainis2011quantum} proof that enable proving $\widetilde\Omega(N^{1/3})$ lower bounds rather than $\widetilde\Omega(N^{1/4})$, are the result of several conversations with the mentioned AI system, which proposed these arguments developed in this paper. 
The authors subsequently verified, simplified, and organized the argument, and rewrote its exposition. Important concepts and ideas were manually factored out and written down in a (hopefully) intuitive, understandable way. The authors take full responsibility for the correctness and content of the paper.

\newpage

\bibliographystyle{alpha}

\bibliography{reference}

\begin{thebibliography}{}

\end{thebibliography}


\newcommand{\etalchar}[1]{$^{#1}$}
\begin{thebibliography}{CFMdW10}

\bibitem[AA23]{DBLP:conf/www/AdriaensA23}
Florian Adriaens and Simon Apers.
\newblock Testing cluster properties of signed graphs.
\newblock In {\em {WWW}}, pages 49--59. {ACM}, 2023.

\bibitem[ACL11]{ambainis2011quantum}
Andris Ambainis, Andrew~M Childs, and Yi-Kai Liu.
\newblock Quantum property testing for bounded-degree graphs.
\newblock In {\em International Workshop on Approximation Algorithms for Combinatorial Optimization}, pages 365--376. Springer, 2011.

\bibitem[Ape20]{DBLP:journals/quantum/Apers20}
Simon Apers.
\newblock Expansion testing using quantum fast-forwarding and seed sets.
\newblock {\em Quantum}, 4:323, 2020.

\bibitem[AS04]{aaronson2004quantum}
Scott Aaronson and Yaoyun Shi.
\newblock Quantum lower bounds for the collision and the element distinctness problems.
\newblock {\em Journal of the ACM (JACM)}, 51(4):595--605, 2004.

\bibitem[BBC{\etalchar{+}}01]{beals2001quantum}
Robert Beals, Harry Buhrman, Richard Cleve, Michele Mosca, and Ronald de~Wolf.
\newblock Quantum lower bounds by polynomials.
\newblock {\em Journal of the ACM (JACM)}, 48(4):778--797, 2001.

\bibitem[BCG{\etalchar{+}}20]{Ben-DavidCGKPW20}
Shalev Ben{-}David, Andrew~M. Childs, Andr{\'{a}}s Gily{\'{e}}n, William Kretschmer, Supartha Podder, and Daochen Wang.
\newblock Symmetries, graph properties, and quantum speedups.
\newblock In Sandy Irani, editor, {\em 61st {IEEE} Annual Symposium on Foundations of Computer Science, {FOCS} 2020, Durham, NC, USA, November 16-19, 2020}, pages 649--660. {IEEE}, 2020.

\bibitem[BFNR03]{buhrman2003quantum}
Harry Buhrman, Lance Fortnow, Ilan Newman, and Hein R{\"o}hrig.
\newblock Quantum property testing.
\newblock In {\em Proceedings of the fourteenth annual ACM-SIAM symposium on Discrete algorithms}, pages 480--488, 2003.

\bibitem[BY22]{bhattacharyya2022property}
Arnab Bhattacharyya and Yuichi Yoshida.
\newblock {\em Property Testing: Problems and Techniques}.
\newblock Springer Nature, 2022.

\bibitem[CAH24]{DBLP:conf/tqc/ChenAH24}
Kuo{-}Chin Chen, Simon Apers, and Min{-}Hsiu Hsieh.
\newblock (quantum) complexity of testing signed graph clusterability.
\newblock In {\em {TQC}}, volume 310 of {\em LIPIcs}, pages 8:1--8:16. Schloss Dagstuhl - Leibniz-Zentrum f{\"{u}}r Informatik, 2024.

\bibitem[CFMdW10]{chakraborty2010new}
Sourav Chakraborty, Eldar Fischer, Arie Matsliah, and Ronald de~Wolf.
\newblock New results on quantum property testing.
\newblock In {\em 30th International Conference on Foundations of Software Technology and Theoretical Computer Science}, page 145, 2010.

\bibitem[CS10]{DBLP:conf/propertytesting/CzumajS10}
Artur Czumaj and Christian Sohler.
\newblock Sublinear-time algorithms.
\newblock In {\em Property Testing - Current Research and Surveys}. 2010.

\bibitem[FMSS03]{friedl2003quantum}
Katalin Friedl, Fr{\'e}d{\'e}ric Magniez, Miklos Santha, and Pranab Sen.
\newblock Quantum testers for hidden group properties.
\newblock In {\em International Symposium on Mathematical Foundations of Computer Science}, pages 419--428. Springer, 2003.

\bibitem[Gol17]{goldreich2017introduction}
Oded Goldreich.
\newblock {\em Introduction to property testing}.
\newblock Cambridge University Press, 2017.

\bibitem[GR99]{goldreich99sublinear}
Oded Goldreich and Dana Ron.
\newblock A sublinear bipartiteness tester for bounded degree graphs.
\newblock {\em Combinatorica}, 19:335--373, 1999.

\bibitem[GR02]{goldreich2002property}
Oded Goldreich and Dana Ron.
\newblock Property testing in bounded degree graphs.
\newblock {\em Algorithmica}, 32:302--343, 2002.

\bibitem[GR11]{goldreich2011testing}
Oded Goldreich and Dana Ron.
\newblock On testing expansion in bounded-degree graphs.
\newblock {\em Studies in Complexity and Cryptography}, pages 68--75, 2011.

\bibitem[Har53]{Harary53}
Frank Harary.
\newblock On the notion of balance of a signed graph.
\newblock {\em Michigan Mathematical Journal}, 2(2):143--146, 1953.

\bibitem[HLW06]{HLW06}
Shlomo Hoory, Nathan Linial, and Avi Wigderson.
\newblock Expander graphs and their applications.
\newblock {\em Bulletin of the American Mathematical Society}, 43(4):439--561, October 2006.

\bibitem[MdW16]{montanaro2016survey}
Ashley Montanaro and Ronald de~Wolf.
\newblock A survey of quantum property testing.
\newblock {\em Theory Comput.}, 7:1--81, 2016.

\bibitem[NC00]{nielsen_chuang}
Michael~A. Nielsen and Isaac~L. Chuang.
\newblock {\em Quantum Computation and Quantum Information}.
\newblock Cambridge University Press, Cambridge, UK, 2000.

\bibitem[NS10]{Nachmias_shapira_expansion}
Asaf Nachmias and Asaf Shapira.
\newblock Testing the expansion of a graph.
\newblock {\em Information and Computation}, 208(4):309--314, 2010.

\bibitem[Pat92]{paturi1992degree}
Ramamohan Paturi.
\newblock On the degree of polynomials that approximate symmetric boolean functions (preliminary version).
\newblock In {\em Proceedings of the twenty-fourth annual ACM symposium on Theory of computing}, pages 468--474, 1992.

\bibitem[Ron09]{DBLP:journals/fttcs/Ron09}
Dana Ron.
\newblock Algorithmic and analysis techniques in property testing.
\newblock {\em Foundations and Trends in Theoretical Computer Science}, 2009.

\bibitem[RS11]{DBLP:journals/siamdm/RubinfeldS11}
Ronitt Rubinfeld and Asaf Shapira.
\newblock Sublinear time algorithms.
\newblock {\em SIAM Journal on Discrete Mathematics (SIDMA)}, 2011.

\bibitem[Sho94]{Shor1994}
Peter~W. Shor.
\newblock Algorithms for quantum computation: Discrete logarithms and factoring.
\newblock In {\em Proceedings of the 35th Annual Symposium on Foundations of Computer Science (FOCS)}, pages 124--134. IEEE, 1994.

\bibitem[Zas18]{Zaslavsky18}
Thomas Zaslavsky.
\newblock Negative (and positive) circles in signed graphs: A problem collection.
\newblock {\em AKCE International Journal of Graphs and Combinatorics}, 15(1):31--48, 2018.

\end{thebibliography}

\appendix
\renewcommand{\theequation}{A.\arabic{equation}}
\setcounter{equation}{1}

\section{Remaining proofs} \label{sec:appendix}
\subsection{Proof of \Cref{lem:coefficient-bound}}
\label{app:proof_2.8}
\lemcoefficientbound*

\begin{proof}
The first part of the claim was proved in \cite[Lemma 15]{ambainis2011quantum}: let $S$ be a set of variables that appear in the polynomial, and let $c_S$ denote its coefficient, then $|c_S|\le2^k$.

    To finish the proof, notice that there are $3N$ possible queries and $O(N)$ non-null responses per query. Therefore the number of canonical monomials of degree $k$ is at most $(CN^2)^k$ for an absolute constant $C$. Then
\[
 \sum_S\abs{c_S}
 \leq\sum_{k=0}^{2q}(2CN^2)^k
 =\exp(O(q\log N)).
\]
\end{proof}

\subsection{Proof of \cref{lem:topology-power}}
\label{app:proof_3.8}
\lemtopologypower*

\begin{proof}
Let $
C_1,\ldots,C_c $
be the connected components of \(F\). For each \(i\in[c]\), let
$
v_i := |V(C_i)|,$
and, for each color \(j\in\{1,2,3\}\), let \(d_{i,j}\) denote the number
of color-\(j\) edges in \(C_i\). Thus $
\sum_{i=1}^c v_i = v$ and $\sum_{i=1}^c \sum_{j=1}^3 d_{i,j} = e$.

For \(d\geq 0\), let us define
\[
R_d(M,\ell)
:=
\prod_{a=1}^{d}
\frac{\ell}{M-(2a-1)\ell},
\qquad
R_0(M,\ell):=1.
\]
Conditioned on a collection of vertices being in the same hidden
block, \(R_d(M,\ell)\) is the probability that \(d\) prescribed disjoint
edges occur in one fixed perfect matching.

We now describe how the connected components of \(F\) may be distributed
among the hidden blocks. Let \(L\) be a partition of \([c]\). A class
\(S\in L\) represents a set of connected components that are placed in
the same hidden block. Define
\begin{equation}
\label{15}
f_L(M,\ell)
:=
\prod_{S\in L}
\prod_{j=1}^3
R_{\sum_{i\in S} d_{i,j}}(M,\ell).
\end{equation}
Thus, conditioned on all components indexed by a class \(S\) lying in
one common hidden block, \(f_L(M,\ell)\) is the probability that all
prescribed matching edges of \(F\) occur.

We use the refinement order on partitions: \(L'\leq L\) means that
\(L'\) is a refinement of \(L\). For \(L'\leq L\), let \(c_{L',L}\)
denote the corresponding inclusion--exclusion coefficient on the
partition lattice; equivalently, \(c_{L,L}=1\), and for \(L'<L\),
\(c_{L',L}\) is the signed number of refinement chains from \(L'\) to
\(L\). Define
\begin{equation}
\label{16}
f'_L(M,\ell)
:=
\sum_{L'\leq L}
c_{L',L}\,f_{L'}(M,\ell).
\end{equation}

We next express the topology moment \(W_F(M,\ell)\) in terms of these
quantities. For a class \(S\in L\), the probability that all vertices in
the components \(\{C_i:i\in S\}\) choose one common hidden block is $\ell^{-\left(\sum_{i\in S}v_i\right)+1}$.
Hence, for the whole partition \(L\), this probability is
$
\prod_{S\in L}
\ell^{-\left(\sum_{i\in S}v_i\right)+1}
=
\ell^{-v+|L|}.$
By the partition inclusion--exclusion formula used in
\cite[Lemma~3]{ambainis2011quantum}, we therefore have
\begin{equation}
\label{17}
W_F(M,\ell)
=
\sum_L
\ell^{-v+|L|}
f'_L(M,\ell).
\end{equation}

We now track the powers of \(\ell\) in each summand. From the definition of \(R_d\),
\begin{equation}
\label{18}
R_d(M,\ell)
=
\ell^d R'_d(M,\ell),
\qquad
R'_d(M,\ell)
:=
\prod_{a=1}^{d}
\frac{1}{M-(2a-1)\ell}.
\end{equation}
Since \(F\) contains \(e\) prescribed colored edges in total, every
\(f_L(M,\ell)\), and hence every \(f'_L(M,\ell)\), contains the explicit
factor $
\ell^e.$

More precisely, let us define \(\widetilde f_L\) and
\(\widetilde f'_L\) by replacing every \(R_d\) in
\cref{15}--\cref{16} by \(R'_d\). Then $f_L(M,\ell)=\ell^e\widetilde f_L(M,\ell)$
and consequently
\begin{equation}
\label{19}
f'_L(M,\ell)=\ell^e\widetilde f'_L(M,\ell).
\end{equation}

We now use the divisibility statement proved in
\cite[Proposition~6]{ambainis2011quantum}. More precisely, the proof of that
proposition shows that, after the explicit factors \(\ell^d\) have been
removed from the terms \(R_d\), the numerator of
\(\widetilde f'_L(M,\ell)\), when written over a denominator consisting
only of factors of the form $
M-(2a-1)\ell, $
is divisible by $
\ell^{c-|L|}.$
Therefore there exist polynomials \(G_L(M,\ell)\) and denominators
\(Q_L(M,\ell)\), where \(Q_L\) is a product of factors
\(M-(2a-1)\ell\), such that
\[
\widetilde f'_L(M,\ell)
=
\frac{
\ell^{c-|L|}G_L(M,\ell)
}{
Q_L(M,\ell)
}.
\]
Using \cref{19}, we obtain
\begin{equation}
\label{20}
f'_L(M,\ell)
=
\frac{
\ell^{e+c-|L|}G_L(M,\ell)
}{
Q_L(M,\ell)
}.  
\end{equation}

Substituting \cref{20} into the \(L\)-summand of \cref{17} gives
\begin{equation}
\ell^{-v+|L|}
f'_L(M,\ell)
=
\frac{
\ell^{-v+|L|}
\ell^{e+c-|L|}
G_L(M,\ell)
}{
Q_L(M,\ell)
}. 
\end{equation}
The dependence on \(|L|\) cancels: $
-v+|L|+e+c-|L|
=
e-v+c$.
Hence,
\begin{equation}
    \label{21}
\ell^{-v+|L|}
f'_L(M,\ell)
=
\frac{
\ell^{e-v+c}G_L(M,\ell)
}{
Q_L(M,\ell)
}.
\end{equation}

Thus every term in the sum \cref{17} has a numerator divisible by $
\ell^{e-v+c}.$ By \Cref{lem:acl-rational}, all the denominators \(Q_L(M,\ell)\) divide the common
denominator
\[
D_q(M,\ell)
=
\prod_{a=1}^{2q}
\left(M-(2a-1)\ell\right)^{\lfloor 2q/a\rfloor}.
\]
Putting all terms in \cref{17} over this common denominator therefore
gives a polynomial \(B_F(M,\ell)\) such that
\[
W_F(M,\ell)
=
\frac{
\ell^{e-v+c}B_F(M,\ell)
}{
D_q(M,\ell)
}.
\]
Moreover, by the degree bound in \Cref{lem:acl-rational}, the numerator has
degree \(O(q\log q)\), and therefore $\deg B_F = O(q\log q)$.

Finally, since $\beta(F)=e-v+c$,
we conclude that
\[
W_F(M,\ell)
=
\frac{
\ell^{\beta(F)}B_F(M,\ell)
}{
D_q(M,\ell)
}.
\]
\end{proof}

\subsection{Proof of \Cref{prop:formal-gap}}
\label{app:proof_3.10}

\propformalgap*

\begin{proof}
By Lemma~\ref{lem:beals}, the acceptance probability of the \(q\)-query quantum
algorithm can be written, after putting the response indicators in
canonical multilinear form, as a real polynomial
\[
P(X)=\sum_{\mu} \alpha_\mu m_\mu(X)
\]
of degree at most \(2q\), where each \(m_\mu\) is a canonical monomial
in the oracle-response indicators.

For \(\sigma\in\{\mathrm{yes},\mathrm{no}\}\), let
\[
p^{\mathrm{form}}_{\sigma}(M,\ell)
:=
\mathbb{E}^{\mathrm{form}}_{\sigma,M,\ell}[P(X)]
\]
denote the formal average acceptance probability. Hence
\[
\Delta^{\mathrm{form}}(M,\ell)
=
p^{\mathrm{form}}_{\YES}(M,\ell)
-
p^{\mathrm{form}}_{\NO}(M,\ell)
\]
satisfies
\begin{equation}
\label{eq:formal-gap-expansion}  
\Delta^{\mathrm{form}}(M,\ell)
=
\sum_{\mu}\alpha_\mu
\left(
\mathbb{E}^{\mathrm{form}}_{\YES,M,\ell}[m_\mu]
-
\mathbb{E}^{\mathrm{form}}_{\NO,M,\ell}[m_\mu]
\right).  
\end{equation}

We analyze each monomial separately.

If \(m_\mu\) is infeasible, then its expectation is zero under both
formal distributions, so it contributes nothing to
\cref{eq:formal-gap-expansion}. We may therefore restrict attention
to feasible monomials.

Fix such a monomial \(m_\mu\). By the canonical form described after
Lemma~\ref{lem:beals}, \(m_\mu\) determines a consistent colored multigraph pattern
\(F_\mu\), together with a prescribed sign on each of its distinct
edges. Write $
v_\mu:=v(F_\mu),\ 
e_\mu:=e(F_\mu),\ 
c_\mu:=c(F_\mu)$,
and $
\beta_\mu
:=
\beta(F_\mu)
=
e_\mu-v_\mu+c_\mu$.

Let
$W_{F_\mu}(M,\ell)$
be the formal topology moment from Definition~\ref{def:topology-moment}; that is, the
probability that all the colored edges prescribed by \(F_\mu\) occur,
ignoring their signs.

Conditional on this topology, the YES and NO distributions differ
only in the probability of the prescribed sign pattern. Define
\[
s^{\NO}_\mu
:=
\Pr_{\NO}
[\text{the prescribed signs on }F_\mu]
\quad\text{and}\quad
s^{\YES}_\mu
:=
\Pr_{\YES}
[\text{the prescribed signs on }F_\mu].
\]
These quantities depend only on the fixed signed pattern
\(F_\mu\), and not on \(M\) or \(\ell\). Therefore
\[
\mathbb{E}^{\mathrm{form}}_{\NO,M,\ell}[m_\mu]
=
W_{F_\mu}(M,\ell)s^{\NO}_\mu
\quad\text{and}\quad
\mathbb{E}^{\mathrm{form}}_{\YES,M,\ell}[m_\mu]
=
W_{F_\mu}(M,\ell)s^{\YES}_\mu .
\]
Consequently,
\begin{equation}
\mathbb{E}^{\mathrm{form}}_{\YES,M,\ell}[m_\mu]
-
\mathbb{E}^{\mathrm{form}}_{\NO,M,\ell}[m_\mu]
=
\delta_\mu W_{F_\mu}(M,\ell),
\label{eq:monomial-gap}
\end{equation}
where $\delta_\mu
:=
s^{\YES}_\mu-s^{\NO}_\mu$.

We now distinguish two cases. First suppose that \(F_\mu\) is a forest. By Lemma~\ref{lem:forest-agreement}, for every
prescribed signing of a forest the $\YES$ and $\NO$ sign probabilities are
equal. Hence $s^{\YES}_\mu=s^{\NO}_\mu$,
so $\delta_\mu=0$.
Thus every forest monomial contributes exactly zero to
\(\Delta^{\mathrm{form}}(M,\ell)\).

Now suppose that \(F_\mu\) is not a forest. Then $\beta_\mu=\beta(F_\mu)\geq 1$.
By \Cref{lem:topology-power}, there exists a polynomial
\(B_{F_\mu}(M,\ell)\) of degree \(O(q\log q)\) such that
\begin{equation}
W_{F_\mu}(M,\ell)
=
\frac{
\ell^{\beta_\mu}B_{F_\mu}(M,\ell)
}{
D_q(M,\ell)
}.
\label{eq:topology-moment-use}
\end{equation}
Combining \cref{eq:monomial-gap} and
\cref{eq:topology-moment-use}, the contribution of \(m_\mu\) to the
formal acceptance gap is
\[
\alpha_\mu\delta_\mu
W_{F_\mu}(M,\ell)
=
\frac{
\alpha_\mu\delta_\mu
\ell^{\beta_\mu}B_{F_\mu}(M,\ell)
}{
D_q(M,\ell)
}.
\]
Since \(\beta_\mu\geq 1\), we may factor one power of \(\ell\):
\begin{equation}
\alpha_\mu\delta_\mu
W_{F_\mu}(M,\ell)
=
\frac{\ell}{D_q(M,\ell)}
\left(
\alpha_\mu\delta_\mu
\ell^{\beta_\mu-1}
B_{F_\mu}(M,\ell)
\right).
\label{eq:one-ell-factor}
\end{equation}

Since all forest monomials vanish, summing
\cref{eq:one-ell-factor} over the feasible nonforest monomials gives
\[
\Delta^{\mathrm{form}}(M,\ell)
=
\frac{\ell}{D_q(M,\ell)}
\sum_{\substack{\mu:\,m_\mu\text{ feasible}\\
                 F_\mu\text{ not a forest}}}
\alpha_\mu\delta_\mu
\ell^{\beta_\mu-1}
B_{F_\mu}(M,\ell).
\]
Let us define
\begin{equation}
H(M,\ell)
:=
\sum_{\substack{\mu:\,m_\mu\text{ feasible}\\
                 F_\mu\text{ not a forest}}}
\alpha_\mu\delta_\mu
\ell^{\beta_\mu-1}
B_{F_\mu}(M,\ell).
\label{eq:H-definition}
\end{equation}
Then \(H(M,\ell)\) is a polynomial and
\[
\Delta^{\mathrm{form}}(M,\ell)
=
\frac{\ell H(M,\ell)}{D_q(M,\ell)}.
\]

It remains only to verify the degree bound. Since \(F_\mu\) arises
from a monomial of degree at most \(2q\), it has at most \(2q\) distinct
prescribed edges, and therefore $\beta_\mu\leq e_\mu\leq 2q$.
Moreover, \Cref{lem:topology-power} gives $\deg B_{F_\mu}=O(q\log q)$.
Hence every summand in \cref{eq:H-definition} has degree at most
\[
(\beta_\mu-1)+\deg B_{F_\mu}
\leq
2q+O(q\log q)
=
O(q\log q).
\]
Taking a finite sum does not increase the maximum degree, so $\deg H=O(q\log q)$.
This proves the proposition.
\end{proof}

\subsection{Proof of \Cref{lem:actual-formal}}\label{app:actual-formal}

\lemactualformal*

\begin{proof}
Let $P(X)=\sum_{\mu} a_\mu m_\mu(X)$ be the canonical acceptance polynomial of the \(q\)-query algorithm.
By Lemma~\ref{lem:coefficient-bound}, for some absolute constant \(C>0\), $\sum_\mu |a_\mu| \le \exp(Cq\log N)$. Let us fix \(\sigma\in\{\YES,\NO\}\). Note that when no overflow event happens, the formal and feasible distributions are similar, so we can couple them in a natural way.
Hence, for every canonical
monomial \(m_\mu\),
\[
|\mathbb{E}^{\mathrm{feas}}_{\sigma}[m_\mu]
-
\mathbb{E}^{\mathrm{form}}_{\sigma}[m_\mu]|
\le
\Pr[\mathrm{overflow}].
\]
From \cref{lem:overflow}, we know that for every \((M,\ell)\in \mathcal G_N\),
$\Pr[\mathrm{overflow}]
\le \exp(-cN^{1/3})
$ for some absolute constant \(c>0\). Therefore
\begin{align*}
|p^{\mathrm{feas}}_{\sigma}(M,\ell)
-
p^{\mathrm{form}}_{\sigma}(M,\ell)|
&\le
\sum_\mu |a_\mu|
\left|\left(\mathbb{E}^{\mathrm{feas}}_{\sigma}[m_\mu]
-
\mathbb{E}^{\mathrm{form}}_{\sigma}[m_\mu]\right)\right|
 \\
&\le
\exp(Cq\log N-cN^{1/3}).
\end{align*}

Now we set
\[
\eta(M,\ell)
:=
\Delta(M,\ell)-\Delta^{\mathrm{form}}(M,\ell).
\]
Since
$\Delta
=p^{\mathrm{feas}}_{\YES}
-p^{\mathrm{feas}}_{\NO}
$
and
$ \Delta^{\mathrm{form}}
= p^{\mathrm{form}}_{\YES}
- p^{\mathrm{form}}_{\NO},
$ using
the triangle inequality, we have
\[
|\eta(M,\ell)|
\le
2\exp(Cq\log N-cN^{1/3}).
\]

Finally, if
$q\le c_0\frac{N^{1/3}}{\log N}
$ for a sufficiently small absolute constant \(c_0>0\), then
\[
Cq\log N-cN^{1/3}
\le
-(c-Cc_0)N^{1/3}
=
-\Omega(N^{1/3}),
\]
and hence
$
\eta(M,\ell)=o(1)
$ holds for every parameter pair in \(\mathcal G_N\).
\end{proof}

\subsection{Proof of \Cref{lem:reciprocal}}
\label{app:proof_3.14}
\lemreciprocal*

\begin{proof}
For simplicity, in the proof we  will abbreviate the functions' notations and only write $D, A, z_i$.

For each $i\in[d]$, the assumption $\abs{z_i}\leq 1/4<1$
allows us to use the geometric-series identity
\begin{equation}\label{eq:single-geometric-series}
 \frac{1}{1-z_i}
 =
 \sum_{k_i=0}^{\infty}z_i^{k_i}.
\end{equation}
This series is absolutely convergent, since
$\sum_{k_i=0}^{\infty}\abs{z_i}^{k_i} = \frac{1}{1-\abs{z_i}}
\leq \frac{4}{3}$.
Consequently, the product of the $d$ series is also absolutely
convergent. More explicitly,
\[
 \sum_{k_1,\ldots,k_d\geq 0}
 \abs{z_1^{k_1}\cdots z_d^{k_d}}
 =
 \prod_{i=1}^{d}
 \left(
   \sum_{k_i=0}^{\infty}\abs{z_i}^{k_i}
 \right)
 \leq
 \left(\frac{4}{3}\right)^d
 <\infty.
\]
Absolute convergence permits us to multiply the series in
\cref{eq:single-geometric-series} and regroup the resulting
monomials according to their total degree. Hence
\begin{equation}\label{eq:product-geometric-expansion}
 \prod_{i=1}^{d}(1-z_i)^{-1}
 =
 \sum_{k_1,\ldots,k_d\geq 0}
 z_1^{k_1}\cdots z_d^{k_d}
 =
 \sum_{j=0}^{\infty}S_j(z_1,\ldots,z_d),
\end{equation}
where
\begin{equation}\label{eq:Sj-definition-detailed}
 S_j(z_1,\ldots,z_d)
 :=
 \sum_{\substack{
       k_1,\ldots,k_d\geq 0\\
       k_1+\cdots+k_d=j}}
 z_1^{k_1}\cdots z_d^{k_d}.
\end{equation}
Thus $S_j$ is the sum of all monomials of total degree $j$ in
the auxiliary variables $z_1,\ldots,z_d$. Since each $z_i$ is
affine in the original variables, $S_j$ has degree at most $j$
in those variables.

{

While considering the polynomial $D$, we will consider $d=d_q$ and enumerate the index set $\mathcal I_q$ (\cref{eq:index-set-IK}) in any order. The polynomial in \cref{eq:SKj} can be represented as

$$S_{q,j}(M, \ell) = S_j((z_i(M, \ell)_{i \in \mathcal I_q})).$$

Since $S_j$ is symmetric, it is invariant under the ordering. Note that $S_{q,0}=1$ and $\deg_{M, \ell}S_{q,j} \leq j$.

}

We next bound $S_j$ uniformly on the region $\mathcal R$.
The number of nonnegative integer solutions of
$k_1+\cdots+k_d=j$ is 
$\binom{d+j-1}{j}$. Moreover, every monomial in \cref{eq:Sj-definition-detailed}
satisfies
$\abs{z_1^{k_1}\cdots z_d^{k_d}}
\leq \left(\frac14\right)^{k_1+\cdots+k_d} =4^{-j}$.
Therefore, by the triangle inequality, we can say that
\begin{equation}\label{eq:Sj-first-bound}
 \abs{S_j}
 \leq
 \binom{d+j-1}{j}4^{-j}.
\end{equation}
Using the elementary estimate $\binom{d+j-1}{j} \leq 2^{d+j-1}$, we obtain
\begin{equation}\label{eq:Sj-geometric-bound}
 \abs{S_j}
 \leq
 2^{d+j-1}4^{-j}
 =
 2^{d-j-1}.
\end{equation}

Let us now set
$J:=4d+10$ and consider the truncated polynomial
\begin{equation}\label{eq:A-detailed-definition}
 A
 :=
 \frac{1}{C}\sum_{j=0}^{J}S_j.
\end{equation}
Because $\deg S_j\leq j$, we have
$ \deg A\leq J=4d+10$.

It remains to bound the approximation error. By
\cref{eq:product-geometric-expansion} and
\cref{eq:Sj-geometric-bound}, we can say that
\begin{equation}\label{eq:reciprocal-tail-bound}
\left|\prod_{i=1}^{d}(1-z_i)^{-1}
 -\sum_{j=0}^{J}S_j\right|
 \leq \sum_{j>J}\abs{S_j} \leq
 \sum_{j>J}2^{d-j-1}
 = 2^{d-J-1}.    
\end{equation}

We now convert the additive approximation in
\cref{eq:reciprocal-tail-bound} into the desired relative
approximation. From $\abs{z_i}\leq 1/4$,
$\abs{1-z_i} \leq 1+\abs{z_i} \leq
\frac54$,
and hence we have
\begin{equation}\label{eq:product-upper-bound}
 \prod_{i=1}^{d}\abs{1-z_i}
 \leq
 \left(\frac54\right)^d.
\end{equation}
Using the fact that
$D=C\prod_{i=1}^{d}(1-z_i)
$
and the definition of $A$, we have
\begin{align*}
 DA-1
 &=
 \prod_{i=1}^{d}(1-z_i)
 \left(
   \sum_{j=0}^{J}S_j
   -
   \prod_{i=1}^{d}(1-z_i)^{-1}
 \right).
\end{align*}
Therefore, by \cref{eq:reciprocal-tail-bound} and
\cref{eq:product-upper-bound}, we have:
$$ \abs{DA-1} \leq
 \left(\frac54\right)^d
 2^{d-J-1} =
 2^{-J-1}
 \left(\frac52\right)^d.$$

Substituting $J=4d+10$ gives us
$ 2^{-J-1}\left(\frac52\right)^d
= 2^{-4d-11}\left(\frac52\right)^d
\leq 10^{-2}.$
Thus $\abs{DA-1}\leq 10^{-2}$ holds throughout $\mathcal R$.

All estimates above use only the assumption that
$\abs{z_i}\leq 1/4$ for every $i \in [d]$. Therefore the same polynomial $A$ works
simultaneously at every point of the region $\mathcal R$.    
\end{proof}

\subsection{Proof of \Cref{lem:denominator-approx}}
\label{app:proof_3.15}
\lemdenominatorapprox*

\begin{proof}

Let \begin{equation}\label{eq:index-set-IK_app}
 \mathcal I_q
 :=\set{(a,r):1\leq a\leq 2q,\ 1\leq r\leq\lfloor 2q/a\rfloor}.
\end{equation}
Then $\abs{\mathcal I_q}=d_q$, like in \cref{eq:den-degree}. Consider
$i=(a,r)\in\mathcal I_q$.
Since $1\leq a\leq 2q$, $1\leq\ell\leq R$, and
$|M-M_0|\leq R^{3/2}$, we have
\begin{equation}\label{eq:zi-uniform-bound}
|z_i(M,\ell)|=
\frac{|M_0-M+(2a-1)\ell|}{M_0}
\leq \frac{|M-M_0|+(2a-1)\ell}{M_0}\leq
\frac{R^{3/2}+(4q-1)R}{M_0}.    
\end{equation}

We bound the two terms in the numerator separately. From the
definition of $R$, we can say that
\[
    R^{3/2}
    \leq
    \left(\frac{N^{2/3}}{100}\right)^{3/2}
    =
    \frac{N}{1000}.
\]
Moreover, $M_0\geq\frac{11N}{10}$. Consequently, we have
\begin{equation}
    \frac{R^{3/2}}{M_0}
    \leq
    \frac{N/1000}{11N/10}
    =
    \frac1{1100}.
    \label{eq:M-variation-bound}
\end{equation}

For the second term, we can write
\begin{align}
\frac{(4q-1)R}{M_0}
    &\leq
    \frac{4qR}{M_0}
    \leq
    \frac{
        4q\,(N^{2/3}/100)
    }{
        11N/10
    }
    =
    \frac{2}{55}
    \frac{q}{N^{1/3}} 
    = o(1),
    \label{eq:ell-shift-bound}
\end{align}
where the last equality follows from
$q=o(N^{1/3})$.

Combining
\cref{eq:zi-uniform-bound},
\cref{eq:M-variation-bound}, and
\cref{eq:ell-shift-bound}, we obtain
\[
    \sup_{\substack{
        i\in\mathcal I_q\\
        1\leq\ell\leq R\\
        |M-M_0|\leq R^{3/2}}}
    |z_i(M,\ell)|
    \leq
    \frac1{1100}+o(1).
\]
In particular, for all sufficiently large $N$,
$|z_i(M,\ell)|\leq\frac14
$
simultaneously for every $i\in\mathcal I_q$ and every
$(M,\ell)$ in the stated region.

We may therefore apply Lemma~\ref{lem:reciprocal}
with $C=M_0^{d_q}$, $d=d_q$,
and with the repeated affine functions
$\bigl(z_i(M,\ell)\bigr)_{i\in\mathcal I_q}$.
The total-degree truncation constructed in that lemma is
precisely
\[
    M_0^{-d_q}
    \sum_{j=0}^{4d_q+10}S_{q,j}(M,\ell),
\]
which is the polynomial $A_q(M,\ell)$ defined above. Hence
    $\left|
        D_q(M,\ell)A_q(M,\ell)-1
    \right|
    \leq 10^{-2}$
holds throughout the entire region.

It remains to bound the degree. We have
$    d_q=
    \sum_{a=1}^{2q}
    \left\lfloor\frac 2qa\right\rfloor \leq
    2q\sum_{a=1}^{2q}\frac1a
    \leq
    2q(1+\log 2q).
    \label{eq:dK-harmonic-bound}
$
Therefore
$
    J_q
    =
    4d_q+10
    =
    O\!\bigl(q\log(q)\bigr).
$
Since every $S_{q,j}$ has degree at most $j$ in $(M,\ell)$, we have 
    $\deg A_q
    \leq J_q
    =
    O\!\bigl(q\log(q)\bigr)$.
This proves the lemma.
\end{proof}

\subsection{Proof of \cref{lem:logdegree-expansion}}
\label{app:proof_4.2}
\lemlogdegreeexpansion*

\begin{proof}
Let us condition on the positions of the $N$ visible vertices among the $M_0$ ambient vertices. Consider a set $U$ of $u\leq N/2$ visible vertices and consider one uniformly random perfect matching. We repeatedly choose an unmatched vertex of $U$ and expose its partner, for $\lceil u/2\rceil$ steps. Such a vertex is always available, since each exposure removes at most two vertices of $U$.

Before the $(t+1)$-th exposure, at most $t$ visible vertices outside $U$ have been used. Hence at least $N-u-t\geq N/4$
unmatched visible vertices outside $U$ remain. The exposed partner is uniform among the remaining ambient vertices. Since $M_0\leq6N/5$ for sufficiently large $N$, the conditional probability that this exposure gives a crossing edge is at least $1/5$.

The number of crossing edges exposed in the $c$ independent matchings therefore stochastically dominates
\[
 Z\sim\Bin\!\left(c\left\lceil\frac u2\right\rceil,\frac15\right).
\]
The mean of the random variable $Z$ is at least $cu/10$. Using Chernoff bound, we have
\[
 \Prb\!\left[|e_H(U,V\setminus U)|<\frac{cu}{20}\right]
 \leq \exp\!\left(-\frac{cu}{80}\right).
\]
Taking a union bound over all sets $U$ gives
\[
 \sum_{u=1}^{\lfloor N/2\rfloor}
 \binom Nu\exp\!\left(-\frac{cu}{80}\right)
 \leq
 \sum_{u=1}^{\lfloor N/2\rfloor}
 \left(\frac{eN}{u}N^{-5}\right)^u=o(1),
\]
where we used $c \geq400\log N$. This proves the lemma.
\end{proof}

\subsection{Proof of \cref{lem:replacement-expansion}}
\label{app:proof_4.5}
\lemreplacementexpansion*

\begin{proof}
Let $S\subseteq V(\mathcal R(H))$ and assume $t:=\abs S\leq Nc/2$. For each cloud let us define $S_u=S\cap C_u$, and define $U=\set{u:\abs{S_u}>c/2}$ and
 $S_0=\bigcup_{u\in U}C_u$.
Let $b=\abs{S\triangle S_0}
  =\sum_u\min\set{\abs{S_u},c-\abs{S_u}}$.
  
By the expansion of the cloud graph, the number of internal edges crossing the cut is at least
\begin{equation}\label{eq:internal-replacement-bound}
 |e_{\mathrm{int}}(S,S^c)|\geq h_0b.
\end{equation}
For the external edges, let us start from the cut determined by the union of whole clouds $S_0$. Changing the cut-membership of one replacement-graph vertex changes the status of at most one external edge. Hence
\[
 |e_{\mathrm{ext}}(S,S^c)|
 \geq |e_H(U,U^c)|-b.
\]
Moreover, both $\abs{S_0}=c\abs U$ and its complement have size at least $t-b$. Therefore
\begin{equation}\label{eq:external-replacement-bound}
 |e_{\mathrm{ext}}(S,S^c)|
 \geq \beta(t-b)-b
 =\beta t-(\beta+1)b.
\end{equation}

Now if $b\geq\beta t/(2(\beta+1))$, then \cref{eq:internal-replacement-bound} is at least $\frac{h_0\beta}{2(\beta+1)}t$.
Otherwise, \cref{eq:external-replacement-bound} is at least $\beta t/2$. Thus we have
\[
 |e_{\mathcal R(H)}(S,S^c)|\geq
 \rho t,
 \qquad
 \rho:=\min\set{\frac\beta2,
                   \frac{h_0\beta}{2(\beta+1)}}>0.
\]
Since the maximum degree is four, every outside boundary vertex is incident to at most four crossing edges. Hence $\abs{\partial_{\mathcal R(H)}S}
 \geq \frac{\rho}{4}\abs S$.
Choosing $\alpha=\rho/4$, we have the proof.
\end{proof}

\subsection{Proof of \Cref{lem:no-subdivision}}\label{app:proofof_4.11}

\lemnosubdivision*

\begin{proof}
Let $V_1,V_2$ be the two classes of visible vertices corresponding to
the two hidden blocks in the construction of $H$. Similar to the proof of \cref{lem:replacement-no-farness}, let us define
\[
    S_i := \bigcup_{u\in V_i} C_u,
    \qquad i\in\{1,2\}.
\]
Since no block overflows, we have
$|V_i|\ge N-\frac{M_0}{2} \ge \frac{9N}{20}-2.$
Consequently, we have
\begin{equation}\label{eq:no-subdivision-original-side-size}
|S_i| = c|V_i| \ge \frac{9Nc}{20}-2c= \frac{9\overline n}{20}-2c.   
\end{equation}
Moreover, every edge of $G$ has both endpoints in either $S_1$ or
$S_2$, and hence there is no edge between $S_1$ and $S_2$.

Every edge that is subdivided therefore has both endpoints in one of
the sets $S_i$, $i \in [2]$. Let us assign its new subdivision vertex to that same set.
Let $Z_i$ be the set of subdivision vertices assigned to $S_i$, and let us write
\[
    S_i' := S_i\cup Z_i, \qquad i\in\{1,2\}.
\]
Then $V(G')=S_1'\mathbin{\dot\cup}S_2'$ and
$|e_{G'}(S_1',S_2')|=0$.
Let $ n' := |V(G')|=\overline n+r$.
Since $r\le \overline n/4$, we have
$\frac{n'}{3} \le \frac{5\overline n}{12}.
$
On the other hand, since $\overline n=Nc$, for all sufficiently large $N$, $\frac{9\overline n}{20}-2c \ge \frac{5\overline n}{12}.
$
This implies that
$\frac{\overline n}{30}\ge 2c$,
which holds whenever $N\ge 60$. It follows from
\cref{eq:no-subdivision-original-side-size} that
\[
    |S_i'|
    \ge |S_i|
    \ge \frac{n'}{3},
    \qquad i\in\{1,2\}.
\]

Let $S$ be the smaller of $S_1'$ and $S_2'$. Then we have
\[
    \frac{n'}{3}
    \le |S|
    \le \frac{n'}{2},
    \qquad
    \partial_{G'}(S)=\varnothing.
\]
Suppose that $\widetilde G$ is an $\alpha'$-vertex-expander obtained
from $G'$ by edge modifications. In $\widetilde G$, the fixed set $S$
must have at least
$|\partial_{\widetilde G}(S)| \ge \alpha'|S|\ge \frac{\alpha'n'}{3}$
distinct outside neighbors. Initially, $S$ has no outside neighbor.
Each added edge creates at most one new outside neighbor of $S$, while
edge deletions create none. Therefore at least
$\frac{\alpha'n'}{3}$ edge modifications are necessary.

Moreover, subdividing edges does not increase the maximum degree, so $G'$ has
maximum degree at most four. Finally,  dividing the required number of
modifications by the bounded-degree normalization $4n'$ gives us
$\frac{\alpha'n'/3}{4n'}= \frac{\alpha'}{12}$.
Thus $G'$ is $\alpha'/12$-far from every
$\alpha'$-vertex-expander. This completes the proof.
\end{proof}

\subsection{Proof of \Cref{lem:bounded-subdivision}}\label{app:proofof_4.12}

\lemboundedsubdivision*

\begin{proof}
Let $m=\abs{V(G)}$,  $r\leq m/4$ be the number of subdivision vertices, and let $G'$ be the subdivided graph. Consider $S\subseteq V(G')$ with $\abs S\leq(m+r)/2$. Let us define $A=S\cap V(G)$ and let $B$ be the subdivision vertices in $S$.

Every original edge crossing $(A,V(G)\setminus A)$ is replaced by a path whose endpoints lie on opposite sides, and therefore gives at least one crossing edge of $G'$. Hence the cut has at least
\begin{equation}\label{eq:subdivision-first-bound}
 \rho\min\set{\abs A,m-\abs A}
\end{equation}
crossing edges. In addition, the vertices in $B$ have $2\abs B$ incident path edges, of which at most $4\abs A$ end in $A$. Thus the cut also has at least
\begin{equation}\label{eq:subdivision-second-bound}
 \max\set{0,2\abs B-4\abs A}
\end{equation}
crossing edges.

Suppose first that $\abs A\geq\abs S/6$. If $\abs A\leq m/2$, then \cref{eq:subdivision-first-bound} is at least $\rho\abs S/6$. If $\abs A>m/2$, then
\[
 m-\abs A\geq m-\abs S
 \geq \frac{m-r}{2}
 \geq \frac{3m}{8}.
\]
Since $\abs S\leq(m+r)/2\leq5m/8$, this is at least $3\abs S/5$, and the first bound is again a constant multiple of $\abs S$.

Finally, if $\abs A<\abs S/6$, then $\abs B>5\abs S/6$, so the expression in \cref{eq:subdivision-second-bound} is larger than $\abs S$. Thus $G'$ has constant edge expansion. Since the maximum degree is still four, dividing the number of crossing edges by four gives constant vertex expansion. This completes the proof.
\end{proof}

\end{document}